\documentclass[11pt, a4paper,oneside]{article}
  \usepackage{amsthm}
  \usepackage{amsmath}
  \usepackage{amssymb}
  \usepackage{secdot}
  \usepackage{fontawesome}
  \usepackage{cite}
  \usepackage{tabularx}
  \usepackage{rotating}
  \usepackage{enumitem}
  \usepackage[dvipsnames]{xcolor}
  \usepackage{setspace}
  \usepackage[left=1in, right=1in, top=1.2in, bottom=1.2in]{geometry}
  \usepackage{titlesec}
  \usepackage{xpatch}
  \usepackage{ctable}
  \usepackage[hang,flushmargin]{footmisc} 
  \newtheorem{theorem}{Theorem}[section]
  
  \newtheorem{lemma}[theorem]{Lemma}
  \newtheorem{definition}{Definition}[section]
  \newtheorem{example}[definition]{Example}

  \renewenvironment{proof}{{\noindent\bf{Proof }}}{\hfill $\blacksquare$\\}
  \titleformat{\section}
  {\normalfont\large\bfseries}{\thesection}{0.4em}{}
  \titleformat{\subsection}
  {\normalfont\bfseries}{\thesubsection}{0.4em}{}
  \makeatletter
  \AtBeginDocument{\xpatchcmd{\@thm}{\thm@headpunct{.}}{\thm@headpunct{}}{}{}}
  \makeatother
  \allowdisplaybreaks
\begin{document}
		\noindent \rule{\textwidth}{0.5mm}\\
	\vspace{0.1cm}\\
	\noindent\textbf{\Large Linear codes with few weights over $\mathbb{F}_{p}+u\mathbb{F}_{p}$}\\
	\vspace{0.1cm}\\
	\textbf{Pavan Kumar$^{\boldsymbol{1,2},\ast}$ and Noor Mohammad Khan$^{\boldsymbol{1}}$}\\
	{$^{1}$Department of Mathematics, Aligarh Muslim University, Aligarh-202002, India}\\
	{$^{2}$Department of Electronic Systems Engineering, Indian Institute of Science, Bengaluru 560012, India}\\
	\let\thefootnote\relax\footnote{\noindent$^{\ast}$Corresponding author\\  Pavan Kumar\\ pavan4957@gmail.com\\  \vspace*{0.1cm}\\ Noor Mohammad Khan\\nm\_khan123@yahoo.co.in}
	\vspace{1.5cm}\\
	\textbf{\large Abstract}\\
	For any positive integer $m$ and an odd prime $p$; let $\mathbb{F}_{q}+u\mathbb{F}_{q}$, where $q=p^{m}$, be a  ring extension of the ring $\mathbb{F}_{p}+u\mathbb{F}_{p}.$
	In this paper, we construct linear codes over $\mathbb{F}_{p}+u\mathbb{F}_{p}$ by using trace function defined on  $\mathbb{F}_{q}+u\mathbb{F}_{q}$ and determine their Hamming weight distributions by employing symplectic-weight distributions of their Gray images.
\section{Introduction}\label{sec1}
Throughout this paper, let  $p$ be an odd prime and let $\mathbb{F}_{q}$ be the finite field with $q=p^m$ elements for any positive integer $m$. Denote by $\mathbb{F}_{q}^{*}=\mathbb{F}_{q}\setminus\{0\}$,   the multiplicative group of $ \mathbb{F}_{q}$. An $(n, M)$ code over $\mathbb{F}_{p}$ is a subset of $\mathbb{F}_{p}^{n}$ of size $M$.  An $[n, k, d]$ linear code $\mathcal {C}$ is a $k$-dimensional subspace of $ \mathbb{F}_{p}^{n} $ with minimum Hamming-distance $d$. The vectors in a linear code $\mathcal{C}$ are known as \emph{codewords}. It is well known that linear codes are easier to describe, encode and decode than nonlinear codes. Therefore, they have been an interesting topic in both theory and practice for many years. 
Moreover, the linear codes  over finite fields  have been studied extensively due to their applications in consumer electronics, data storage systems and communication systems. A lot of linear codes with a few weights applicable in secret sharing \cite{ADHK98, CDY05}, association schemes \cite{CG84} and authentication codes \cite{DW05} have been constructed  by many researchers.
The number of nonzero coordinates in $ c\in \mathcal{C} $ is called the Hamming-weight $\text{wt}(c)$ of a codeword $c$. Let $A_{i}$ denote the number of codewords with Hamming weight $i$ in a linear code $\mathcal{C}$ of length $n$. The weight enumerator of a code $\mathcal{C}$ is a polynomial defined by
$$1+A_1z+A_2z^2+\cdots+A_nz^n,$$
where $(1, A_{1}, \ldots, A_{n})$ is called the \emph{weight distribution}  of $\mathcal{C}$. Throughout the paper, for any set $S$, $\#S$  denotes the cardinality of the set $S$. If $ \#\{i:A_i\neq0, 1\leq i\leq n\}=t$, then the code $\mathcal{C}$ is said to be a $t$-weight code. Several classes of linear codes with various weights have been constructed \cite{D09,DD14,DD15,DKS01,DY13, ZD14,LCXM18,SB12}.

Let $D=\{d_{1},d_{2},\ldots,d_{n}\}\subset\mathbb{F}_{q}^{*}$. A linear code $\mathcal{C}_{D}$ of length $n$ over $\mathbb{F}_{p}$ is defined by
\begin{equation}\label{eq1}
\mathcal{C}_{D}=\{(\mathrm{Tr}(xd_{1}),\mathrm{Tr}(xd_{2}),\ldots,\mathrm{Tr}(xd_{n})):x\in\mathbb{F}_{q}\},
\end{equation}
where $\mathrm{Tr}$ denotes the absolute trace function from $\mathbb{F}_{q}$ onto $\mathbb{F}_{p}$. The set $D$ is called the defining set of this code $\mathcal{C}_{D}$. This construction was proposed by Ding et al. \cite{DD15}. By choosing defining set properly, some optimal linear codes with few weights can be obtained (see\cite{D15,DD14,DD15,LYL14,LWL18,WDX15}).

Modules over finite rings entered in the coding theory in 1990's with  a wide range of applications from number theory (unimodular lattices \cite{A98}) to engineering (low correlation sequences \cite{BGO07}) and to combinatorics (designs \cite{AM74}). In the field of codes over finite rings, the results of Nechaev \cite{N89} and of Hammons et al. \cite{HKCSS94} showed  a significant improvement. Their results settle that the linear codes over the ring $\mathbb{Z}_{4}$  are closely related to nonlinear Kerdock \cite{K17} and Preparata \cite{P68} binary codes.  It has been proved that the codes  over finite rings have more codewords than the ones over finite fields with the same length. The length of the Preparata code is same as of the extended double error-correcting BCH code but Preparata code has twice as many codewords in BCH \cite{HKCSS94}. At present, mostly researchers focus on cyclic codes over  finite rings(see \cite{S09} and the references therein). 

Let  $\mathcal{R}_{m}=\mathbb{F}_{q}+u\mathbb{F}_{q}$ be a ring extension of $\mathcal{R}(=\mathbb{F}_{p}+u\mathbb{F}_{p})$ with degree $m$. A linear code $\mathcal{C}$ of length $n$ over $\mathcal{R}_{m}=\mathbb{F}_{q}+u\mathbb{F}_{q}$ is an $\mathcal{R}_{m}$-submodule of $\mathcal{R}_{m}^{n}$. We can find linear codes of length $n$ over $\mathcal{R}$ based on construction in \eqref{eq1}, but in this case $D=\{d_{1},d_{2},\ldots,d_{n}\}\subset\mathcal{R}_{m}$, and $\mathrm{tr}$ is a linear function from $\mathcal{R}_{m}$ to $\mathcal{R}$.

In this paper, we  investigate the Hamming-weight distributions of the linear codes over the finite ring $\mathbb{F}_{p}+u\mathbb{F}_{p}$ by employing symplectic weight distributions of their Gray images. Some research over weight distributions of the linear codes over  $\mathcal{R}$  based on construction in \eqref{eq1} has been done for some special cases \cite{P68,SLS16,LL19}.

 The complete paper is organized as follows: Section \ref{sec2} recalls some definitions and results on group characters and Gauss sums over finite fields; Section \ref{sec3} completes the ground work to establish  main results; Section \ref{sec4} have the main results and  Section \ref{sec5} concludes this paper.
\section{Preliminaries}\label{sec2}

Let $u$ be an indeterminant such that $\mathbb{F}_{p}+u\mathbb{F}_{p}$ forms a ring under usual addition and multiplication.
For an odd prime $p$ and a positive integer $m>2$, we construct the ring extension $\mathcal{R}_{m}=\mathbb{F}_{q}+u\mathbb{F}_{q}$ of $\mathcal{R}(=\mathbb{F}_{p}+u\mathbb{F}_{p})$ with degree $m$. The \emph{Trace function} from $\mathbb{F}_{q}+u\mathbb{F}_{q}$ onto $\mathbb{F}_{p}+u\mathbb{F}_{p}$ is defined as
$$\mathrm{tr}=\sum_{j=o}^{m-1}F^{j},$$
where $F$ is the Frobenius operator which maps $a+ub$ to $a^{p}+ub^{p}$.
For all $a, b\in\mathbb{F}_{q}$, it can easily be checked that
$\mathrm{tr}(a+ub)=\mathrm{Tr}(a)+u\mathrm{Tr}(b),$  where $\mathrm{Tr}$ stands for the absolute trace function from $\mathbb{F}_{q}$ onto $\mathbb{F}_{p}$.

 The Gray map $\phi$ from $\mathcal{R}^{n}$ to $\mathbb{F}_{p}^{2n}$ is defined by
$$ \phi(a_{1}+ub_{1},a_{2}+ub_{2},\ldots,a_{n}+ub_{n})=(a_{1},a_{2},\ldots,a_{n},b_{1},b_{2},\ldots,b_{n}),$$
for each $r=(a_{1}+ub_{1},a_{2}+ub_{2},\ldots,a_{n}+ub_{n})\in\mathcal{R}^{n}$. It is a one to one map from $R^{n}$ to $\mathbb{F}_{p}^{2n}$. Symplectic-weight of a vector $(\boldsymbol{a}|\boldsymbol{b})$ in $\mathbb{F}_{p}^{2n}$ is defined by
$$\mathrm{swt}((\boldsymbol{a}|\boldsymbol{b})) = \#\{k:(a_{k},b_{k})\neq(0,0)\}. $$ It is easy to see that \emph{Hamming-weight} of $\textbf{c}\in R^{n}$ is equal to the \emph{symplectic-weight} of the Gray image $\phi(\textbf{c})$.
The \emph{Hamming-distance} of $x, y\in R^{n}$ is defined as $w_{H}(x-y)$ while the symplectic distance $d_{S}$ of $\textbf{x}, \textbf{y}\in \mathbb{F}_{p}^{2n}$ is defined as swt(\textbf{x}-\textbf{y}). Thus the Gray map is, by construction, a linear isometry from $(R^{n}, d{_H}) $ to $(\mathbb{F}_{p}^{2n}, d_{S})$.

Let $a$ be a primitive element of $\mathbb{F}_{q}$ and let $q=Nh+1$ for some positive integers $N>1$ and $h>1$. The \emph{cyclotomic classes} of order $N$ in $\mathbb{F}_{q}$ are the cosets $\mathcal{C}_{i}^{(N, q)}=a^{i}\langle a^{N}\rangle$ for $i=0,1,\ldots,N-1$, where $\langle a^{N}\rangle$ denotes the subgroup of $\mathbb{F}_{q}^{*}$ generated by $a^{N}$. It is obvious that $\#\mathcal{C}_{i}^{(N,q)}=h$. For fixed $i$ and $j$, the \emph{cyclotomic number} $(i,j)^{(N,q)}$ is defined to be the number of solutions of the equation $$x_{i}+1=x_{j},\quad x_{i}\in\mathcal{C}_{i}^{(N,q)},x_{j}\in\mathcal{C}_{j}^{(N,q)},$$
where $1=a^{0}$ is the multiplicative identity of $\mathbb{F}_{q}$. That is, $(i,j)^{(N,q)}$ is the number of ordered pairs $(s,t)$ such that
$$a^{Ns+i}+1=a^{Nt+j},\quad0\leq s,t\leq h-1.$$

Now, we present some notions and results about group characters and Gauss sums  for later use (see \cite{LN97} for details).
An additive character $\chi$ of $\mathbb{F}_{q}$ is a mapping from $\mathbb{F}_{q}$
into the multiplicative group of complex numbers of absolute value 1 with $\chi\ (g_{1} + g_{2})=\chi\ (g_{1})\chi(g_{2})$ for all  $ g_{1},  g_{2}\in \mathbb{F}_{q}$.\\
By (\cite{LN97}, Theorem 5.7), for any $ b\in  \mathbb{F}_{q} $,
\begin{equation}\label{2.1}
\chi_{b}(x)= \zeta_{p}^{ \mathrm{Tr}(bx)},\quad \text{for all }x\in \mathbb{F}_{q},
\end{equation}
defines an additive character of $ \mathbb{F}_{q} $, where $\zeta_{p}=e^{\frac{2\pi\sqrt{-1}}{p}}$, and every additive character can be obtained in this way. An additive character defined by $ \chi_{0}(x)=1 $ for all $ x\in \mathbb{F}_{q}$ is called the trivial character and all other characters are called nontrivial characters. The character $ \chi_{1} $ in \eqref{2.1} is called the canonical additive character of $ \mathbb{F}_{q}$. The orthogonal property of additive characters of $ \mathbb{F}_{q} $ can be found in (\cite{LN97}, Theorem 5.4) and is given as
\begin{equation}
\sum_{x\in \mathbb{F}_{q}}\chi(x)=
\begin{cases}
q, &\text{if $\chi$ trivial};\\
0, &\text{if $\chi$ non-trivial}.
\end{cases}
\end{equation}
\\
Characters of the multiplicative group $ \mathbb{F}^{*}_{q}$  of $\mathbb{F}_{q}  $ are called multiplicative characters of $ \mathbb{F}_{q} $. By (\cite{LN97}, Theorem 5.8), for each $j=0, 1, \ldots, q-2$, the function $ \psi_{j} $ with
$$\psi_{j}(g^{k})=e^{\frac{2\pi\sqrt{-1}jk}{q-1}}\quad \text{for}\  k=0, 1,\ldots, q-2 
$$
defines a multiplicative character of $ \mathbb{F}_{q} $, where $g$ is a generator of $ \mathbb{F}^{*}_{q} $. For $ j=\frac{q-1}{2} $, the character $ \eta=\psi_{\frac{q-1}{2}} $ defined by\\
$$\eta(g^{k})=
\begin{cases}
-1, &\text{if  $ 2\nmid k$};\\
~~1, &\text{if $ 2\mid k$}. 
\end{cases}
$$
is called the quadratic character. Moreover, we extend this quadratic character by letting $ \eta(0)=0 $.
The quadratic Gauss sum $ G=G(\eta, \chi_{1}) $ over $ \mathbb{F}_{q} $ is defined by\
$$ G(\eta, \chi_{1})=\sum_{x\in \mathbb{F}_{q}^{*}}\eta(x)\chi_{1}(x).
$$
Now, let  $ \overline{\eta} $ and $ \overline{\chi}_{1} $ denote the quadratic and canonical character of $ \mathbb{F}_{p} $ respectively. Then we define  the quadratic Gauss sum $ \overline{G}=G(\overline{\eta}, \overline{\chi}_{1}) $ over $ \mathbb{F}_{p} $ by
$$ G(\overline{\eta}, \overline{\chi}_{1})=\sum_{x\in \mathbb{F}_{p}^{*}}\overline{\eta}(x)\overline{\chi}_{1}(x).$$
The explicit values of quadratic Gauss sums are given by the following lemma.
\begin{lemma}\label{p4le1}
	\emph{\cite[Theorem 5.15]{LN97}} Let the symbols have the same meanings as before. Then\\
	
	$G(\eta, \chi_{1})=(-1)^{m-1}\sqrt{-1}^{\frac{(p-1)^{2}m}{4}}\sqrt{p^m},~G(\overline{\eta}, \overline{\chi}_{1})=\sqrt{-1}^{\frac{(p-1)^{2}}{4}}\sqrt{p}.$ 
\end{lemma}
\begin{lemma}\label{p4le2}
	\emph{\cite[Lemma 7]{DD15}} Let $\eta$ and $\overline{\eta}$ be the  quadratic characters of $\mathbb{F}_{q}^{*}$ and $\mathbb{F}_{p}^{*}$, respectively. Then the following assertions hold:\\
	$1$. If $m\geq2$ is even, then $ \eta(y)=1 $ for each $y\in \mathbb{F}^{*}_{p}  $.\\
	$2$. If $m$ is odd, then $ \eta(y)=\overline{\eta}(y) $ for each $y\in \mathbb{F}^{*}_{p}  $. 
\end{lemma}
\begin{lemma}{\em\cite{S67}}\label{p4le3}
Let the notations have the same significations as before. Then, for $N=2$, the cyclotomic numbers are given by:\\
	$1.$ h even: $(0,0)^{(2,q)}=\frac{h-2}{2}$, $(0,1)^{(2,q)}=(1,0)^{(2,q)}=(1,1)^{(2,q)}=\frac{h}{2}$.\\
	$2.$ h odd: $(0,0)^{(2,q)}=(1,0)^{(2,q)}=(1,1)^{(2,q)}=\frac{h-1}{2}$, $(0,1)^{(2,q)}=\frac{h+1}{2}.$
\end{lemma}
\begin{lemma}\label{p4le4}
	\emph{\cite[Theorem 5.33]{LN97}} Let $ \chi $ be a non-trivial additive character of $ \mathbb{F}_{q} $, and let $f(x)=a_{2}x^{2} + a_{1}x + a_{0}\in \mathbb{F}_{q}[x]$ with $ a_{2}\neq 0 $. Then
	$$\sum_{x\in \mathbb{F}_{q}}\chi(f(x))=\chi(a_{0} - a_{1}^{2}(4a_{2})^{-1})\eta(a_{2})G(\eta, \chi).$$
\end{lemma}
\begin{lemma}\label{p4le5}
	\emph{\cite[Lemma 9]{DD15}}
	For each $c\in\mathbb{F}_{p},$ let
	$n_{c}=\#\{v\in\mathbb{F}_{q}:\mathrm{Tr}(v^{2})=c\}.$ Then
	$$n_{c}=\begin{cases}
		p^{m-1}+p^{-1}(p-1)G,&\text{if } 2\mid m\text{ and } c=0;\\
		p^{m-1}-p^{-1}G,&\text{if } 2\mid m\text{ and } c\neq0;\\
			p^{m-1},&\text{if } 2\nmid m\text{ and } c=0;\\
		p^{m-1}+p^{-1}\overline{\eta}(-c)G\overline{G},&\text{if } 2\nmid m\text{ and } c\neq0.
	\end{cases}$$
\end{lemma}
\begin{lemma}\label{p4le6}
	\emph{\cite[Lemma 10]{DD15}}
	For any $\xi\in\mathbb{F}_{q}^{*}$, let
	$$K_{\xi}=\sum_{d\in\mathbb{F}_{q}}\sum_{x\in\mathbb{F}_{p}^{*}}\zeta_{p}^{x\mathrm{Tr}(d^{2})}\sum_{z\in\mathbb{F}_{p}^{*}}\zeta_{p}^{z\mathrm{Tr}(\xi d)}.$$ Then	
$$K_{\xi}=\begin{cases}
		G(p-1)^{2},&\text{if } 2\mid m \text{ and } \mathrm{Tr}(\xi^{2})=0;\\
		-G(p-1),&\text{if } 2\mid m \text{ and } \mathrm{Tr}(\xi^{2})\neq0;\\
		0,&\text{if } 2\nmid m \text{ and } \mathrm{Tr}(\xi^{2})=0;\\
		\overline{\eta}(-\mathrm{Tr}(\xi^{2}))(p-1)G\overline{G},&\text{if } 2\nmid m \text{ and } \mathrm{Tr}(\xi^{2})\neq0.
	\end{cases}$$
		
\end{lemma}
\begin{lemma}\label{p4le7}
	\emph{\cite[Lemma 11]{DD15}}
	Let $\xi\in\mathbb{F}_{q}^{*}$ and 
	$N_{\xi}=\#\{d\in\mathbb{F}_{q}:\mathrm{Tr}(d^{2})=0 \text{ and }\mathrm{Tr}(\xi d)=0\}.$ Then
	$$N_{\xi}=\begin{cases}
		p^{m-2}+p^{-1}(p-1)G,&\text{if }2\mid m \text{ and }\mathrm{Tr}(\xi^{2})=0;\\
			p^{m-2},&\text{if }2\mid m \text{ and }\mathrm{Tr}(\xi^{2})\neq0;\\
		p^{m-2},&\text{if }2\nmid m \text{ and }\mathrm{Tr}(\xi^{2})=0;\\
	p^{m-2}+p^{-2}\overline{\eta}(-\mathrm{Tr}(\xi^{2}))(p-1)G\overline{G},&\text{if }2\nmid m \text{ and }\mathrm{Tr}(\xi^{2})\neq0.\\

	\end{cases}$$
\end{lemma}
\section{Primary results needed to establish main results}\label{sec3}
In this section, we finish ground work to establish our main results. 
\begin{lemma}\label{p4le8}
	For $c_{1}, c_{2}\in\mathbb{F}_{p}$, let
	$$N(c_{1},c_{2})=\#\{(a, b)\in\mathbb{F}_{q}\times\mathbb{F}_{q} : \mathrm{Tr}(a^{2})=c_{1}~and~\mathrm{Tr}(b^{2})=c_{2}\}.$$
	Then, for odd $m$, we have
	$$N(c_{1}, c_{2})=\begin{cases}
	p^{2m-2},&\text{if $c_{1}=0$ and $c_{2}=0$};\\
	p^{2m-2}+p^{m-2}\overline{\eta}(-c_{1})G\overline{G},&\text{if $c_{1}\neq0$ and $c_{2}=0$};\\
	p^{2m-2}+p^{m-2}\overline{\eta}(-c_{2})G\overline{G},&\text{if $c_{1}=0$ and $c_{2}\neq0$};\\
	p^{2m-2}+p^{m-2}\big(\overline{\eta}(-c_{1})+\overline{\eta}(-c_{2})\big)G\overline{G}+p^{-2}\overline{\eta}(c_{1}c_{2})G^{2}\overline{G}^{2},&\text{if $c_{1}\neq0$ and $c_{2}\neq0$};
	\end{cases}$$
	and, for even $m$, we have
	$$N(c_{1}, c_{2})=\begin{cases}
	(p^{m-1}+p^{-1}(p-1)G)^{2},&\text{if $c_{1}=0$ and $c_{2}=0$};\\
	p^{2m-2}+p^{m-2}(p-2)G-p^{-2}(p-1)G^{2},&\text{if $c_{1}\neq0$ and $c_{2}=0$};\\
	p^{2m-2}+p^{m-2}(p-2)G-p^{-2}(p-1)G^{2},&\text{if $c_{1}=0$ and $c_{2}\neq0$};\\
	(p^{m-1}-p^{-1}G)^{2},&\text{if $c_{1}\neq0$ and $c_{2}\neq0$}.
	\end{cases}~~~~~~~~~~~~~~~~~~~$$
\end{lemma}
\begin{proof}
	The proof follows directly from Lemma \ref{p4le5}.
\end{proof}
\begin{lemma}\label{p4le9}
	 For $i\in\{1,2,3\}$, let $c_{i}\in\mathbb{F}_{p}^{*}$. For $j\in\{1,-1\}$, let $ L_{j} $ denote the number of triplets $ (c_{1},c_{2},c_{3}) $ such that $ \overline{\eta}(c_{3}^{2}-c_{1}c_{2})=j$. Then
	$$L_{j}=\begin{cases}
	\frac{1}{2}(p-1)^{2}(p-3),& \text{if } j=1;\\
	\frac{1}{2}(p-1)^{3},&\text{if } j=-1.
	\end{cases}$$
\end{lemma}
\begin{proof} For each pair $(c_{3},c_{2})\in\mathbb{F}_{p}^{*}\times\mathbb{F}_{p}^{*}$ and $s\in\mathbb{F}_{p}^{*}$ fixed, define a mapping $\mathcal{L}$ from $\mathbb{F}_{p}^{*}\times\mathbb{F}_{p}^{*}$ into $\mathbb{F}_{p}$ by
	$\mathcal{L}(c_{3},c_{2})=c_{3}^{2}-sc_{2}.$
	For each $c_{0}\in\mathbb{F}_{p}^{*}$, let 
	$$\mathcal{A}_{c_{0}}=\{(c_{3},c_{2})\in\mathbb{F}_{p}^{*}\times\mathbb{F}_{p}^{*}: \mathcal{L}(c_{3},c_{2})=c_{0}\}.$$
	Set $p=2h+1$. Now, for a fixed $c_{0}$
	such that $\overline{\eta}(c_{0})=1$, the number of pairs $(c_{3}^{2},c_{2})$ satisfying  $c_{3}^{2}-sc_{2}=c_{0}$ is equal to	$(0,0)^{(2,p)}+(1,0)^{(2,p)}=h-1~(\text{by Lemma \ref{p4le2})}.$
	Similarly, for a fixed $c_{0}$
	such that $\overline{\eta}(c_{0})=-1$, the number of pairs $(c_{3}^{2},c_{2})$ satisfying $c_{3}^{2}-sc_{2}=c_{0}$  is equal to
	$(0,1)^{(2,p)}+(1,1)^{(2,p)}=h~(\text{from Lemma \ref{p4le2})}.$
	Consequently, we have
	\begin{equation*}
	\#\mathcal{A}_{c_{0}}=\begin{cases}
	2(h-1),&\text{ if } \overline{\eta}(c_{0})=1;\\
	2h,&\text{ if } \overline{\eta}(c_{0})=-1.
	\end{cases}
	\end{equation*}
	Since there are $(p-1)$ choices to fix $s$, and there are $\frac{(p-1)}{2}$ choices for $c_{0}$ such that $\overline{\eta}(c_{0})=1$ or $\overline{\eta}(c_{0})=-1$, we can easily conclude that $L_{1}=(p-1)^{2}(h-1)$ and $L_{-1}=(p-1)^{2}h$. Thus the result is established.\end{proof}
\begin{lemma}\label{p4le10}
	For $c_{1},c_{3}\in\mathbb{F}_{p}$, let
	$$\Psi_{2}=\sum_{a,b\in\mathbb{F}_{q}}\sum_{x\in\mathbb{F}_{p}^{*}}\zeta_{p}^{x\mathrm{Tr}(a^{2})-xc_{1}}\sum_{z\in\mathbb{F}_{p}^{*}}\zeta_{p}^{z\mathrm{Tr}(ab)-zc_{3}}.$$ Then
$$\Psi_{2}=\begin{cases}
		p^{m}(p-1)^{2},&\text{if } c_{1}=0 \text{ and } c_{3}=0;\\
		-p^{m}(p-1),&\text{if } c_{1}=0 \text{ and } c_{3}\neq0;\\
		-p^{m}(p-1),&\text{if } c_{1}\neq0 \text{ and } c_{3}=0;\\
		p^{m},&\text{if } c_{1}\neq0 \text{ and } c_{3}\neq0.\\
	\end{cases}$$
\end{lemma}
\begin{proof} By Lemma \ref{p4le4}, we have
	\begin{align*}
	\Psi_{2}&=\sum_{a,b\in\mathbb{F}_{q}}\sum_{x\in\mathbb{F}_{p}^{*}}\zeta_{p}^{x\mathrm{Tr}(a^{2})-xc_{1}}\sum_{z\in\mathbb{F}_{p}^{*}}\zeta_{p}^{z\mathrm{Tr}(ab)-zc_{3}}\\
	&=\sum_{x,z\in\mathbb{F}_{p}^{*}}\overline{\chi}_{1}(-xc_{1}-zc_{3})\sum_{a,b\in\mathbb{F}_{q}}\chi_{1}(xa^{2}+zab)\\
	&=\sum_{x,z\in\mathbb{F}_{p}^{*}}\overline{\chi}_{1}(-xc_{1}-zc_{3})\sum_{b\in\mathbb{F}_{q}}\chi_{1}(-\frac{z^{2}b^{2}}{4x})\eta(x)G\\
	&=G\sum_{x,z\in\mathbb{F}_{p}^{*}}\overline{\chi}_{1}(-xc_{1}-zc_{3})\eta(x)\sum_{b\in\mathbb{F}_{q}}\chi_{1}(-\frac{z^{2}b^{2}}{4x})\\
	&=G\sum_{x,z\in\mathbb{F}_{p}^{*}}\overline{\chi}_{1}(-xc_{1}-zc_{3})\eta(x)\chi_{1}(0)\eta(-\frac{z^{2}}{4x})G\\
	&=\eta(-1)G^{2}\sum_{x,z\in\mathbb{F}_{p}^{*}}\overline{\chi}_{1}(-xc_{1}-zc_{3})\\
	&=\begin{cases}
	\eta(-1)G^{2}(p-1)^{2},&\text{if } c_{1}=0 \text{ and } c_{3}=0;\\
	-\eta(-1)G^{2}(p-1),&\text{if } c_{1}=0 \text{ and } c_{3}\neq0;\\
	-\eta(-1)G^{2}(p-1),&\text{if } c_{1}\neq0 \text{ and } c_{3}=0;\\
	\eta(-1)G^{2},&\text{if } c_{1}\neq0 \text{ and } c_{3}\neq0;\\
	\end{cases}
	\end{align*}
as required.	
\end{proof}
\begin{lemma}\label{p4le11}
	For $c_{1},c_{2},c_{3}\in\mathbb{F}_{p}^{*}$, let
	$$\delta=\sum_{c\in\mathbb{F}_{p}^{*}}\sum_{x\in\mathbb{F}_{p}^{*}}\overline{\chi}_{1}(x(c-c_{1}))\sum_{y\in\mathbb{F}_{p}^{*}}\overline{\chi}_{1}(-yc_{2})\sum_{z\in\mathbb{F}_{p}^{*}}\overline{\chi}_{1 }(-\frac{z^{2}c}{4y}-zc_{3}).$$ Then
	$$\delta=\begin{cases}
	p+1,&\text{ if } c_{3}^{2}-c_{1}c_{2}=0;\\ 
	p^{2}+p+1,&\text{ if } \overline{\eta}(c_{3}^{2}-c_{1}c_{2})=1;\\
	-p^{2}+p+1,&\text{ if } \overline{\eta}(c_{3}^{2}-c_{1}c_{2})=-1.\\
	\end{cases}$$
\end{lemma}
\begin{proof} By Lemma \ref{p4le4}, we have
	{\small\begin{align*}
		\delta&=\sum_{c\in\mathbb{F}_{p}^{*}}\sum_{x\in\mathbb{F}_{p}^{*}}\overline{\chi}_{1}(x(c-c_{1}))\sum_{y\in\mathbb{F}_{p}^{*}}\overline{\chi}_{1}(-yc_{2})\sum_{z\in\mathbb{F}_{p}^{*}}\overline{\chi}_{1 }(-\frac{z^{2}c}{4y}-zc_{3})\\
		&= \sum_{c\in\mathbb{F}_{p}^{*}}\sum_{x\in\mathbb{F}_{p}^{*}}\overline{\chi}_{1}(x(c-c_{1}))\sum_{y\in\mathbb{F}_{p}^{*}}\overline{\chi}_{1}(-yc_{2})\bigg(\overline{\chi}_{1}\big(\frac{yc_{3}^{2}}{c}\big)\overline{\eta}\big(-\frac{c}{4y})\overline{G}-1\bigg)\\
		&=\sum_{c\in\mathbb{F}_{p}^{*}}\sum_{x\in\mathbb{F}_{p}^{*}}\overline{\chi}_{1}(x(c-c_{1}))\sum_{y\in\mathbb{F}_{p}^{*}}\overline{\chi}_{1}(-yc_{2})\overline{\chi}_{1}\big(\frac{yc_{3}^{2}}{c}\big)\overline{\eta}\big(-\frac{c}{4y})\overline{G}-\sum_{c\in\mathbb{F}_{p}^{*}}\sum_{x\in\mathbb{F}_{p}^{*}}\overline{\chi}_{1}(x(c-c_{1}))\sum_{y\in\mathbb{F}_{p}^{*}}\overline{\chi}_{1}(-yc_{2})\\
		&=\sum_{c\in\mathbb{F}_{p}^{*}}\sum_{x\in\mathbb{F}_{p}^{*}}\overline{\chi}_{1}(x(c-c_{1}))\sum_{y\in\mathbb{F}_{p}^{*}}\overline{\chi}_{1}\big(\frac{y(c_{3}^{2}-cc_{2})}{c}\big)\overline{\eta}\big(-\frac{c}{4y})\overline{G}+\sum_{c\in\mathbb{F}_{p}^{*}}\sum_{x\in\mathbb{F}_{p}^{*}}\overline{\chi}_{1}(x(c-c_{1}))\\
		&=\sum_{c\in\mathbb{F}_{p}^{*}}\sum_{x\in\mathbb{F}_{p}^{*}}\overline{\chi}_{1}(x(c-c_{1}))\sum_{y\in\mathbb{F}_{p}^{*}}\overline{\chi}_{1}\big(\frac{y(c_{3}^{2}-cc_{2})}{c}\big)\overline{\eta}\big(-\frac{c}{4y})\overline{G}+\sum_{x\in\mathbb{F}_{p}^{*}}\overline{\chi}_{1}(-xc_{1})\sum_{c\in\mathbb{F}_{p}^{*}}\overline{\chi}_{1}(xc)\\
		&=\overline{G}(p-1)\sum_{y\in\mathbb{F}_{p}^{*}}\overline{\chi}_{1}\big(\frac{y(c_{3}^{2}-c_{1}c_{2})}{c_{1}}\big)\overline{\eta}\big(-\frac{c_{1}}{4y})-\overline{G}\sum_{c\in\mathbb{F}_{p}^{*}\setminus\{c_{1}\}}\sum_{y\in\mathbb{F}_{p}^{*}}\overline{\chi}_{1}\big(\frac{y(c_{3}^{2}-cc_{2})}{c}\big)\overline{\eta}\big(-\frac{c}{4y})+1.
		\end{align*}}
	A simple calculation leads us to the following values of $\delta$:
	$$\delta=\begin{cases}
	p+1,&\text{ if } c_{3}^{2}-c_{1}c_{2}=0;\\ 
	p^{2}+p+1,&\text{ if } \overline{\eta}(c_{3}^{2}-c_{1}c_{2})=1;\\
	-p^{2}+p+1,&\text{ if } \overline{\eta}(c_{3}^{2}-c_{1}c_{2})=-1;
	\end{cases}$$ as required.
\end{proof}
\begin{lemma}\label{p4le12}
	For $c_{1},c_{2},c_{3}\in\mathbb{F}_{p}^{*},$ let
	$$\rho=\sum_{c\in\mathbb{F}_{p}^{*}}\sum_{x\in\mathbb{F}_{p}^{*}}\overline{\chi}_{1}(x(c-c_{1}))\sum_{y\in\mathbb{F}_{p}^{*}}\overline{\chi}_{1}(-yc_{2})\overline{\eta}(y)\sum_{z\in\mathbb{F}_{p}^{*}}\overline{\chi}_{1}(-\frac{z^{2}c}{4y}-zc_{3}).$$ Then
	$$\rho=\begin{cases}
	-\big(\overline{\eta}(-c_{1})+\overline{\eta}(-c_{2})\big)\overline{G}+\overline{\eta}(-c_{1})(p-1)^{2}\overline{G},&\text{if } c_{3}^{2}-c_{1}c_{2}=0;\\
	-(p+1)\overline{\eta}(-c_{2})\overline{G}-p\overline{\eta}(-c_{1})\overline{G},&\text{if } c_{3}^{2}-c_{1}c_{2}\neq0.
	\end{cases}$$
\end{lemma}
\begin{proof}
	{\small\begin{align*}
		\rho&=\sum_{c,x\in\mathbb{F}_{p}^{*}}\overline{\chi}_{1}(x(c-c_{1}))\sum_{y\in\mathbb{F}_{p}^{*}}\overline{\chi}_{1}(-yc_{2})\overline{\eta}(y)\sum_{z\in\mathbb{F}_{p}^{*}}\overline{\chi}_{1}(-\frac{z^{2}c}{4y}-zc_{3})\\
		&=\sum_{c,x\in\mathbb{F}_{p}^{*}}\overline{\chi}_{1}(x(c-c_{1}))\sum_{y\in\mathbb{F}_{p}^{*}}\overline{\chi}_{1}(-yc_{2})\overline{\eta}(y)\bigg(\overline{\chi}_{1}\big(\frac{yc_{3}^{2}}{c}\big)\overline{\eta}\big(\frac{-c}{4y}\big)\overline{G}-1\bigg)\text{ by Lemma $\ref{p4le4}$}\\
		&=\sum_{c,x\in\mathbb{F}_{p}^{*}}\overline{\chi}_{1}(x(c-c_{1}))\sum_{y\in\mathbb{F}_{p}^{*}}\overline{\chi}_{1}\big(\frac{y(c_{3}^{2}-cc_{2})}{c}\big)\overline{\eta}\big(-c)\overline{G}-\sum_{c,x\in\mathbb{F}_{p}^{*}}\overline{\chi}_{1}(x(c-c_{1}))\sum_{y\in\mathbb{F}_{p}^{*}}\overline{\chi}_{1}(-yc_{2})\overline{\eta}(y)\\
		&=\overline{G}\sum_{c,x\in\mathbb{F}_{p}^{*}}\overline{\chi}_{1}(x(c-c_{1}))\sum_{y\in\mathbb{F}_{p}^{*}}\overline{\chi}_{1}\big(\frac{y(c_{3}^{2}-cc_{2})}{c}\big)\overline{\eta}\big(-c)-\overline{\eta}(-c_{2})\overline{G}\sum_{c,x\in\mathbb{F}_{p}^{*}}\overline{\chi}_{1}(x(c-c_{1}))\\
		&=(p-1)\overline{G}\sum_{y\in\mathbb{F}_{p}^{*}}\overline{\chi}_{1}\big(\frac{y(c_{3}^{2}-c_{1}c_{2})}{c_{1}}\big)\overline{\eta}\big(-c_{1})-\overline{G}\sum_{c\in\mathbb{F}_{p}^{*}\setminus\{c_{1}\}}\sum_{y\in\mathbb{F}_{p}^{*}}\overline{\chi}_{1}\big(\frac{y(c_{3}^{2}-cc_{2})}{c}\big)\overline{\eta}\big(-c)-\overline{\eta}(-c_{2})\overline{G}\\
		&=(p-1)\overline{\eta}(-c_{1})\overline{G}\sum_{y\in\mathbb{F}_{p}^{*}}\overline{\chi}_{1}\big(\frac{y(c_{3}^{2}-c_{1}c_{2})}{c_{1}}\big)-\overline{G}\sum_{c\in\mathbb{F}_{p}^{*}\setminus\{c_{1}\}}\overline{\eta}\big(-c)\sum_{y\in\mathbb{F}_{p}^{*}}\overline{\chi}_{1}\big(\frac{y(c_{3}^{2}-cc_{2})}{c}\big)-\overline{\eta}(-c_{2})\overline{G}
		\end{align*}}\\
	From now onward, we divide the remaining proof into two cases.\\
	\textbf{Case 1:} Suppose that $c_{3}^{2}-c_{1}c_{2}=0.$ Then, we have
	{\small\begin{align*}
		\rho&=(p-1)\overline{\eta}(-c_{1})\overline{G}\sum_{y\in\mathbb{F}_{p}^{*}}1+\overline{G}\sum_{c\in\mathbb{F}_{p}^{*}\setminus\{c_{1}\}}\overline{\eta}(-c)-\overline{\eta}(-c_{2})\overline{G}\\
		&=(p-1)^{2}\overline{\eta}(-c_{1})\overline{G}-\overline{\eta}(-c_{1})\overline{G}-\overline{\eta}(-c_{2})\overline{G}\\
		&=-\big(\overline{\eta}(-c_{1})+\overline{\eta}(-c_{2})\big)\overline{G}+\overline{\eta}(-c_{1})(p-1)^{2}\overline{G}.
		\end{align*}}
	\textbf{Case 2:} Consider that $c_{3}^{2}-c_{1}c_{2}\neq0.$ Let $c^{*}\in\mathbb{F}_{p}^{*} $ such that $c_{3}^{2}-c^{*}c_{2}=0.$ Then
	{\small\begin{align*}
		&\rho=-(p-1)\overline{\eta}(-c_{1})\overline{G}-\overline{G}\sum_{c\in\mathbb{F}_{p}^{*}\setminus\{c_{1}\}}\overline{\eta}\big(-c)\sum_{y\in\mathbb{F}_{p}^{*}}\overline{\chi}_{1}\big(\frac{y(c_{3}^{2}-cc_{2})}{c}\big)-\overline{\eta}(-c_{2})\overline{G}\\
		&=-(p-1)\overline{\eta}(-c_{1})\overline{G}-\overline{G}\overline{\eta}(-c^{*})(p-1)-\overline{G}\sum_{c\in\mathbb{F}_{p}^{*}\setminus\{c_{1},c^{*}\}}\overline{\eta}\big(-c)\sum_{y\in\mathbb{F}_{p}^{*}}\overline{\chi}_{1}\big(\frac{y(c_{3}^{2}-cc_{2})}{c}\big)-\overline{\eta}(-c_{2})\overline{G}\\
		&=-(p-1)\overline{\eta}(-c_{1})\overline{G}-\overline{\eta}(-c^{*})(p-1)\overline{G}+\overline{G}\sum_{c\in\mathbb{F}_{p}^{*}\setminus\{c_{1},c^{*}\}}\overline{\eta}\big(-c)-\overline{\eta}(-c_{2})\overline{G}\\
		&=-(p+1)\overline{\eta}(-c_{2})\overline{G}-p\overline{\eta}(-c_{1})\overline{G}.
		\end{align*}}
	This completes the proof.
\end{proof}
\begin{lemma}\label{p4le13}
	For $c_{1},c_{2},c_{3}\in\mathbb{F}_{p}^{*},$ let
	$$\sigma=\sum_{c\in\mathbb{F}_{p}^{*}}\sum_{x\in\mathbb{F}_{p}^{*}}\overline{\chi}_{1}(x(c-c_{1}))\sum_{y\in\mathbb{F}_{p}^{*}}\overline{\chi}_{1}(-yc_{2})\overline{\eta}(-cy)\sum_{z\in\mathbb{F}_{p}^{*}}\overline{\chi}_{1}(-\frac{z^{2}c}{4y}-zc_{3}).$$
	 Then
		$$\sigma=\begin{cases}
	(p^{2}-2p-1)\overline{G},&\text{if } c_{3}^{2}-c_{1}c_{2}=0;\\
	-(p+1)\overline{G}-p\overline{\eta}(c_{1}c_{2})\overline{G},&\text{if } c_{3}^{2}-c_{1}c_{2}\neq0.
	\end{cases}$$   
		\end{lemma}
\begin{proof}
One may easily prove the lemma by using the arguments similar to the arguments used to  prove the  previous lemma.
\end{proof}
\begin{lemma}\label{p4le14}
	For $c_{1},c_{2},c_{3}\in\mathbb{F}_{p}$, let
	$$\Psi_{4}=\sum_{a,b\in\mathbb{F}_{q}}\sum_{x\in\mathbb{F}_{p}^{*}}\zeta_{p}^{x\mathrm{Tr}(a^{2})-xc_{1}}\sum_{y\in\mathbb{F}_{p}^{*}}\zeta_{p}^{y\mathrm{Tr}(b^{2})-yc_{2}}\sum_{z\in\mathbb{F}_{p}^{*}}\zeta_{p}^{z\mathrm{Tr}(ab)-zc_{3}}.$$ Then\\
	$1.$ if $m$ is even, we have
	{\small$$\Psi_{4}=\begin{cases}
	(p-1)^{2}(p^{m}+G(p-2))G,&\text{if }c_{1}=c_{2}=c_{3}=0;\\
	-(p-1)(p^{m}+G(p-2))G,&\text{if exactly one of $c_{i}$ is nonzero};\\
	(p^{m}+G(p-2))G+\overline{\eta}(-c_{1}c_{2})(p^{m+1}-pG)G,&\text{if } c_{1}\neq0,~c_{2}\neq0 \text{ and } c_{3}=0;\\
	G(p^{m}-2G),&\text{if each $c_{i}\neq 0$} \text{ and }c_{3}^{2}-c_{1}c_{2}=0;\\
	G(p+1)(p^{m}-G)-G^{2},&\text{if each $c_{i}\neq0$}  \text{ and }\overline{\eta}(c_{3}^{2}-c_{1}c_{2})=1;\\
	- G(p-1)(p^{m}-G)-G^{2},&\text{if each $c_{i}\neq0$}  \text{ and }\overline{\eta}(c_{3}^{2}-c_{1}c_{2})=-1;\\
	(p^{m}+G(p-2))G,&\text{if } c_{3}\neq0 \text{ and exactly one of }c_{1}\text{ and }c_{2} \text{ is zero}.\\
	\end{cases}$$}
	$2.$ if $m$ is odd, we have
	{\small$$\Psi_{4}=\begin{cases}	
		-p^{-1}(p-1)^{2}G^{2}\overline{G}^{2},&\text{if } c_{1}=c_{2}=c_{3}=0;\\
		p^{-1}(p-1)G^{2}\overline{G}^{2},&\text{if } c_{1}=c_{2}=0\text{ and}~c_{3}\neq0;\\
		(\overline{\eta}(-c_{2})p^{m}+p^{-1}G\overline{G})(p-1)G\overline{G},&\text{if }c_{2}\neq0 \text { and } c_{1}=c_{3}=0;\\
		-(\overline{\eta}(-c_{2})p^{m}+p^{-1}G\overline{G})G\overline{G},&\text{if } c_{1}=0,~c_{2}\neq0 \text { and } c_{3}\neq0;\\
		(\overline{\eta}(-c_{1})p^{m}+p^{-1}G\overline{G})(p-1)G\overline{G},&\text{if } c_{1}\neq0\text{ and }c_{2}=c_{3}=0;\\
		-(\overline{\eta}(-c_{1})p^{m}+p^{-1}G\overline{G})G\overline{G},&\text{if } c_{1}\neq0,~c_{2}=0 \text{ and } c_{3}\neq0;\\	-\big(\overline{\eta}(-c_{1})+\overline{\eta}(-c_{2}) \big)p^{m}G\overline{G}-(\overline{\eta}(c_{1}c_{2})+p^{-1})G^{2}\overline{G}^{2},&\text{if } c_{1}\neq0,~c_{2}\neq0 \text{ and } c_{3}=0;\\
		\overline{\eta}(-c_{1})(p-2)p^{m}G\overline{G} +(p-2-p^{-1})G^{2}\overline{G}^{2},&\text{if each  $c_{i}\neq 0$} \text{ and }c_{3}^{2}-c_{1}c_{2}=0;\\
		-\big(\overline{\eta}(-c_{1})+\overline{\eta}(-c_{2})\big)p^{m}G\overline{G}-(\overline{\eta}(c_{1}c_{2})+1+p^{-1})G^{2}\overline{G}^{2},&\text{if each  $c_{i}\neq 0$} \text{ and }c_{3}^{2}-c_{1}c_{2}\neq0.\end{cases}~~~~$$}
\end{lemma}
\begin{proof} By Lemma \ref{p4le4}, we have
		\begin{align*}
	\Psi_{4}&=\sum_{a,b\in\mathbb{F}_{q}}\sum_{x\in\mathbb{F}_{p}^{*}}\zeta_{p}^{x\mathrm{Tr}(a^{2})-xc_{1}}\sum_{y\in\mathbb{F}_{p}^{*}}\zeta_{p}^{y\mathrm{Tr}(b^{2})-yc_{2   }}\sum_{z\in\mathbb{F}_{p}^{*}}\zeta_{p}^{z\mathrm{Tr}(ab)-zc_{3}}\\
	&=\sum_{x,y,z\in\mathbb{F}_{p}^{*}}\overline{\chi}_{1}(-xc_{1}-yc_{2}-zc_{3})\sum_{a\in\mathbb{F}_{q}}\chi_{1}(xa^{2})\sum_{b\in\mathbb{F}_{q}}\chi_{1}(yb^{2}+zab)\\
	&=\sum_{x,y,z\in\mathbb{F}_{p}^{*}}\overline{\chi}_{1}(-xc_{1}-yc_{2}-zc_{3})\sum_{a\in\mathbb{F}_{q}}\chi_{1}(xa^{2})\chi_{1}\big(-\frac{z^{2}a^{2}}{4y}\big)\eta(y)G\\
	&=G\sum_{x,y,z\in\mathbb{F}_{p}^{*}}\overline{\chi}_{1}(-xc_{1}-yc_{2}-zc_{3})\eta(y)\sum_{a\in\mathbb{F}_{q}}\overline{\chi}_{1}(x\mathrm{Tr}(a^{2})\overline{\chi}_{1}\big(-\frac{z^{2}\mathrm{Tr}(a^{2})}{4y}\big).
	\end{align*}
	Now, we divide the remaining proof into two cases.\\
	\textbf{Case 1:} We consider that $2\mid m$. Then, by Lemmas  \ref{p4le2}, \ref{p4le5} and \ref{p4le11}, we have
	{\small\begin{align*}
	\Psi_{4}&=G(p^{m-1}+p^{-1}(p-1)G)\sum_{x,y,z\in\mathbb{F}_{p}^{*}}\overline{\chi}_{1}(-xc_{1}-yc_{2}-zc_{3})\\&  
	~~~~~~+G(p^{m-1}-p^{-1}G)\sum_{x,y,z\in\mathbb{F}_{p}^{*}}\overline{\chi}_{1}(-xc_{1}-yc_{2}-zc_{3})\sum_{c\in\mathbb{F}_{p}^{*}}\overline{\chi}_{1}(xc)\overline{\chi}_{1 }(-\frac{z^{2}c}{4y})\\
	&=G(p^{m-1}+p^{-1}(p-1)G)\sum_{x,y,z\in\mathbb{F}_{p}^{*}}\overline{\chi}_{1}(-xc_{1}-yc_{2}-zc_{3})\\&  
	~~~~~~+G(p^{m-1}-p^{-1}G)\sum_{x,y\in\mathbb{F}_{p}^{*}}\overline{\chi}_{1}(-xc_{1}-yc_{2})\sum_{c\in\mathbb{F}_{p}^{*}}\overline{\chi}_{1}(xc)\sum_{z\in\mathbb{F}_{p}^{*}}\overline{\chi}_{1 }(-\frac{z^{2}c}{4y}-zc_{3})\\
	&=\begin{cases}
	(p-1)^{2}(p^{m}+G(p-2))G,&\text{if all $c_{i} $ are zero};\\
	-(p-1)(p^{m}+G(p-2))G,&\text{if exactly one of $c_{i}$ is non zero};\\	
	 (p^{m}+G(p-2))G,&\text{if } c_{1}=0,~c_{2}\neq0 \text { and } c_{3}\neq0;\\
	(p^{m}+G(p-2))G,&\text{if } c_{1}\neq0,~c_{2}=0 \text{ and } c_{3}\neq0;\\
	(p^{m}+G(p-2))G+\overline{\eta}(-c_{1}c_{2})(p^{m+1}-pG)G,&\text{if } c_{1}\neq0,~c_{2}\neq0 \text{ and } c_{3}=0;\\
		G(p^{m}-2G),&\text{if each $c_{i}\neq 0$} \text{ and }c_{3}^{2}-c_{1}c_{2}=0;\\
	 G(p+1)(p^{m}-G)-G^{2},&\text{if each $c_{i}\neq0$}  \text{ and }\overline{\eta}(c_{3}^{2}-c_{1}c_{2})=1;\\
		- G(p-1)(p^{m}-G)-G^{2},&\text{if each $c_{i}\neq0$}  \text{ and }\overline{\eta}(c_{3}^{2}-c_{1}c_{2})=-1. 
	\end{cases}
	\end{align*}}
	\textbf{Case 2:} Assume that $2\nmid m$. Then, by Lemmas \ref{p4le2} and \ref{p4le5}, we have
{\small	\begin{align*}
	\Psi_{4}&=p^{m-1}G\sum\limits_{x,y,z\in\mathbb{F}_{p}^{*}}\overline{\chi}_{1}(-xc_{1}-yc_{2}-zc_{3})\overline{\eta}(y)\\&  
	~~~~~~~+G\sum\limits_{x,y,z\in\mathbb{F}_{p}^{*}}\overline{\chi}_{1}(-xc_{1}-yc_{2}-zc_{3})\overline{\eta}(y)\sum\limits_{c\in\mathbb{F}_{p}^{*}}(p^{m-1}+p^{-1}\overline{\eta}(-c)G\overline{G})\overline{\chi}_{1}(xc)\overline{\chi}_{1}(-\frac{z^{2}c}{4y})\\
	&=p^{m-1}G\sum\limits_{x,y,z\in\mathbb{F}_{p}^{*}}\overline{\chi}_{1}(-xc_{1}-yc_{2}-zc_{3})\overline{\eta}(y)\\&~~~~~~+p^{m-1}G\sum_{c\in\mathbb{F}_{p}^{*}}\sum_{x\in\mathbb{F}_{p}^{*}}\overline{\chi}_{1}(x(c-c_{1}))\sum_{y\in\mathbb{F}_{p}^{*}}\overline{\chi}_{1}(-yc_{2})\overline{\eta}(y)\sum_{z\in\mathbb{F}_{p}^{*}}\overline{\chi}_{1}(-\frac{z^{2}c}{4y}-zc_{3})\\&~~~~~~~~~~~~~~+p^{-1}G^{2}\overline{G}\sum_{c\in\mathbb{F}_{p}^{*}}\sum_{x\in\mathbb{F}_{p}^{*}}\overline{\chi}_{1}(x(c-c_{1}))\sum_{y\in\mathbb{F}_{p}^{*}}\overline{\chi}_{1}(-yc_{2})\overline{\eta}(-cy)\sum_{z\in\mathbb{F}_{p}^{*}}\overline{\chi}_{1}(-\frac{z^{2}c}{4y}-zc_{3}).\end{align*}}\\
A simple calculation including the applications of the Lemmas \ref{p4le12} and \ref{p4le13} leads us to the following values of $\Psi_{4}$:
	{\small$$\Psi_{4}=\begin{cases}	
	-p^{-1}(p-1)^{2}G^{2}\overline{G}^{2},&\text{if } c_{1}=0,~c_{2}=0\text{ and}~c_{3}=0;\\
p^{-1}(p-1)G^{2}\overline{G}^{2},&\text{if } c_{1}=0,~c_{2}=0\text{ and}~c_{3}\neq0;\\
	(\overline{\eta}(-c_{2})p^{m}+p^{-1}G\overline{G})(p-1)G\overline{G},&\text{if } c_{1}=0,~c_{2}\neq0 \text { and } c_{3}=0;\\
	-(\overline{\eta}(-c_{2})p^{m}+p^{-1}G\overline{G})G\overline{G},&\text{if } c_{1}=0,~c_{2}\neq0 \text { and } c_{3}\neq0;\\
	(\overline{\eta}(-c_{1})p^{m}+p^{-1}G\overline{G})(p-1)G\overline{G},&\text{if } c_{1}\neq0,~c_{2}=0 \text{ and } c_{3}=0;\\
	-(\overline{\eta}(-c_{1})p^{m}+p^{-1}G\overline{G})G\overline{G},&\text{if } c_{1}\neq0,~c_{2}=0 \text{ and } c_{3}\neq0;\\
	-\big(\overline{\eta}(-c_{1})+\overline{\eta}(-c_{2}) \big)p^{m}G\overline{G}-(\overline{\eta}(c_{1}c_{2})+p^{-1})G^{2}\overline{G}^{2},&\text{if } c_{1}\neq0,~c_{2}\neq0 \text{ and } c_{3}=0;\\
		\overline{\eta}(-c_{1})(p-2)p^{m}G\overline{G} +(p-2-p^{-1})G^{2}\overline{G}^{2},&\text{if each  $c_{i}\neq 0$} \text{ and }c_{3}^{2}-c_{1}c_{2}=0;\\
	-\big(\overline{\eta}(-c_{1})+\overline{\eta}(-c_{2})\big)p^{m}G\overline{G}-(\overline{\eta}(c_{1}c_{2})+1+p^{-1})G^{2}\overline{G}^{2},&\text{if each  $c_{i}\neq 0$} \text{ and }c_{3}^{2}-c_{1}c_{2}\neq0.\end{cases}$$}
	This completes the proof.	
\end{proof}
\begin{lemma}\label{p4le15}
	For $c_{1},c_{2},c_{3}\in\mathbb{F}_{p}$, let
	$$N(c_{1},c_{2},c_{3})=\#\{(a,b)\in\mathbb{F}_{q}:\mathrm{Tr}(a^{2})=c_{1},~\mathrm{Tr}(b^{2})=c_{2}\text{ and }\mathrm{Tr}(ab)=c_{3}\}.$$ 
	Then the following Tables describe $N(c_{1},c_{2},c_{3}).$			
\end{lemma}
\vspace{1cm}
\begin{center}
	\textbf{Table 1:} Value of $N(c_{1},c_{2},c_{3})$ for even $m$			
	\begin{tabular}{|p{8.5 cm}|p{5.5cm}|}				
		\hline
		Value of $N(c_{1},c_{2},c_{3})$& Conditions\\
		\hline
		$p^{2m-3}+p^{m-3}(p-1)(2p-1)+p^{m-3}(p^{2}-1)G+p^{-3}(p-1)^{3}G^{2}$ & $ c_{1}=c_{2}= c_{3}=0$\\
		\hline
		$p^{2m-3}+p^{m-3}(p-1)G-p^{m-3}(2p-1)+p^{-3}(p-1){G}^{2 }$ & $c_{1}=c_{2}=0 \text { and } c_{3}\neq0$\\
		\hline
		$p^{2m-3}-p^{m-3}(p-1)-p^{m-3}G+p^{-3}(p-1)G^{2}+p^{-2}\overline{\eta}(-c_{1}c_{2})(p^{m}-G)G$ & $ c_{1}\neq0,~c_{2}\neq0~\text{and }c_{3}=0$\\
		\hline
		$p^{2m-3}+p^{m-3}(p-1)^{2}-p^{m-3}G-p^{-3}(p-1)^{2}G^{2}$ & $c_{3}=0$ and exactly one of  $c_{1}$ and $c_{2}$ is zero\\
		\hline
		$p^{2m-3}+p^{m-3}(p-1)(G-1)-p^{-3}G^{2}$ & $c_{3}\neq0$ and exactly one of  $c_{1}$ and $c_{2}$ is zero\\				
		\hline
		$p^{2m-3}-p^{m-3}(G-1)-p^{-3}G^{2}$ &$ \text{each }c_{i}\neq 0 \text{ and }c_{3}^{2}-c_{1}c_{2}=0$\\
		\hline
		$p^{2m-3}-p^{m-3}(2G-1)+p^{-3}(p+1)(p^{m}-G)G$ & $\text{each }c_{i}\neq 0 \text{ and }\overline{\eta}(c_{3}^{2}-c_{1}c_{2})=1$\\
		\hline
		$p^{2m-3}-p^{m-3}(2G-1)-p^{-3}(p-1)(p^{m}-G)G$ & $ \text{each }c_{i}\neq 0$ and $\overline{\eta}(c_{3}^{2}-c_{1}c_{2})=-1$\\
		\hline
	\end{tabular}
\end{center}

\begin{center}
	\textbf{Table 2:} Value of $N(c_{1},c_{2},c_{3})$ for odd $m$	\begin{tabular}{|p{8.5cm}| p{5.5cm}|}
		\hline
		Value of $N(c_{1},c_{2},c_{3})$& Conditions\\			
		\hline
		$p^{2m-3}+p^{m-3}(p-1)(2p-1)-p^{-4}(p-1)^{2}G^{2}\overline{G}^{2}$ & $ c_{1}=c_{2}=c_{3}=0$\\
		\hline
		$p^{2m-3}-p^{m-3}(2p-1)+p^{-4}(p-1)  G^{2}\overline{G}^{2}$ & $ c_{1}=c_{2}=0~\text{and}~c_{3}\neq0$\\
		\hline
		$p^{2m-3}+p^{m-3}(p-1)^{2}+p^{-4}(p-1)G^{2}\overline{G}^{2}+p^{m-2}\overline{\eta}(-c_{2})G\overline{G}$ & $ c_{2}\neq0$ and $c_{1}=c_{3}=0$\\
		\hline
		$p^{2m-3}+p^{m-3}(p-1)^{2}+p^{-4}(p-1)G^{2}\overline{G}^{2}+p^{m-2}\overline{\eta}(-c_{1})G\overline{G}$ & $c_{1}\neq0$ and $c_{2}= c_{3}=0$\\
		\hline
		$p^{2m-3}-p^{m-3}(p-1)-p^{-4}G^{2}\overline{G}^{2}$ & $  \text{exactly any two of $c_{i}$ are nonzero }$\\
		\hline
		$p^{2m-3}+p^{m-3}+p^{-4}(p^{2}-p-1)G^{2}\overline{G}^{2}+p^{m-3}(\overline{\eta}(-c_{1})(p-1)+\overline{\eta}(-c_{2}))G\overline{G}$&each $c_{i}\neq0$  and $c_{3}^{2}-c_{1}c_{2}=0$\\
		\hline
		$p^{2m-3}+p^{m-3}-p^{-4}(p+1)G^{2}\overline{G}^{2}$ &each $c_{i}\neq0$  and $c_{3}^{2}-c_{1}c_{2}\neq0$\\
		\hline		
	\end{tabular}
\end{center}
\begin{proof} By the properties of additive character, one can easily obtain that 
	\begin{align*}
	N(c_{1},c_{2},c_{3})&=\frac{1}{p^{3}}\sum_{a,b\in\mathbb{F}_{q}}\Big(\sum_{x\in\mathbb{F}_{p}}\zeta_{p}^{x\mathrm{Tr}(a^{2})-xc_{1}}\Big)\Big(\sum_{y\in\mathbb{F}_{p}}\zeta_{p}^{y\mathrm{Tr}(b^{2})-yc_{2}}\Big)\Big(\sum_{z\in\mathbb{F}_{p}}\zeta_{p}^{z\mathrm{Tr}(ab)-zc_{3}}\Big)\\
	&=\frac{1}{p^{3}}\sum_{a,b\in\mathbb{F}_{q}}\Big(\sum_{x\in\mathbb{F}_{p}}\zeta_{p}^{x\mathrm{Tr}(a^{2})-xc_{1}}\Big)\Big(\sum_{y\in\mathbb{F}_{p}}\zeta_{p}^{y\mathrm{Tr}(b^{2})-yc_{2}}\Big)\Big(1+\sum_{z\in\mathbb{F}_{p}^{*}}\zeta_{p}^{z\mathrm{Tr}(ab)-zc_{3}}\Big)\\
	&=\frac{N(c_{1},c_{2})}{p}+\frac{1}{p^{3}}(\Psi_{1}+\Psi_{2}+\Psi_{3}+\Psi_{4}),
	\end{align*}
	where
	\begin{align*}
	\Psi_{1}&=\sum_{a,b\in\mathbb{F}_{q}}\sum_{z\in\mathbb{F}_{p}^{*}}\zeta_{p}^{z\mathrm{Tr}(ab)-zc_{3}}=\begin{cases}
	p^{m}(p-1),&\text{if } c_{3}=0,\\
	-p^{m},&\text{if } c_{3}\neq0;
	\end{cases}\\
	\Psi_{2}&=\sum_{a,b\in\mathbb{F}_{q}}\sum_{x\in\mathbb{F}_{p}^{*}}\zeta_{p}^{x\mathrm{Tr}(a^{2})-xc_{1}}\sum_{z\in\mathbb{F}_{p}^{*}}\zeta_{p}^{z\mathrm{Tr}(ab)-zc_{3}},\\
	\Psi_{3}&=\sum_{a,b\in\mathbb{F}_{q}}\sum_{y\in\mathbb{F}_{p}^{*}}\zeta_{p}^{y\mathrm{Tr}(b^{2})-yc_{2}}\sum_{z\in\mathbb{F}_{p}^{*}}\zeta_{p}^{z\mathrm{Tr}(ab)-zc_{3}},\\
	\Psi_{4}&=\sum_{a,b\in\mathbb{F}_{q}}\sum_{x\in\mathbb{F}_{p}^{*}}\zeta_{p}^{x\mathrm{Tr}(a^{2})-xc_{1}}\sum_{y\in\mathbb{F}_{p}^{*}}\zeta_{p}^{y\mathrm{Tr}(b^{2})-yc_{2}}\sum_{z\in\mathbb{F}_{p}^{*}}\zeta_{p}^{z\mathrm{Tr}(ab)-zc_{3}}.
	\end{align*}	
	The result now follows directly from the Lemmas \ref{p4le8}, \ref{p4le10} and \ref{p4le14}.
	\end{proof}

From now onward, to avoid the large presentations of the results, we put $\alpha=\mathrm{Tr}(a^{2})$, $\beta=\mathrm{Tr}(b^{2})$ and $\gamma=\mathrm{Tr}(ab)$, where $a,b\in\mathbb{F}_{q}^{*}$.
\begin{lemma}\label{p4le16}
	For  $a, b\in\mathbb{F}_{q}^{*}$, where $a\neq wb$ for all $w\in\mathbb{F}_{p}^{*}$, let
	$$\Omega_{4}=\sum_{d\in\mathbb{F}_{q}}\sum_{x\in\mathbb{F}_{p}^{*}}\zeta_{p}^{x\mathrm{Tr}(d^{2})}\sum_{y\in\mathbb{F}_{p}^{*}}\zeta_{p}^{y\mathrm{Tr}(ad)}\sum_{z\in\mathbb{F}_{p}^{*}}\zeta_{p}^{z\mathrm{Tr}(bd)}.$$ Then\\
	$1.$ if $m$ is even, we have
	$$\Omega_{4}=\begin{cases}
	G(p-1)^{3},& \text{if } \alpha= \beta= \gamma=0;\\
		G(p-1),&\text{if } \alpha\neq0,~ \beta\neq0 \text{ and } \alpha\beta=\gamma^{2}; \\
	\overline{\eta}\big(\gamma^{2}-\alpha\beta\big)p(p-1)G+G(p-1),&\text{if } \alpha\neq0,~ \beta\neq0 \text{ and } \alpha\beta\neq\gamma^{2};\\
	-G(p-1)^{2},& \text{if exactly one of } \alpha, \beta \text{ and } \gamma \text{ is nonzero};\\
		G(p-1),&\text{if exactly one of } \alpha \text{ and }\beta \text{ is zero} \text{ and } \gamma\neq0.
	\end{cases}
	$$
	$2.$ if  $m$ is odd, we have
	$$\Omega_{4}=\begin{cases}
	0,& \text{if } \alpha=\beta=0;\\
	\overline{\eta}(-\beta)(p-1)^{2}G\overline{G},& \text{if } \beta\neq0 \text{ and }\alpha= \gamma=0;\\
	-\overline{\eta}(-\beta)(p-1)G\overline{G},& \text{if } \alpha=0,~ \beta\neq0 \text{ and } \gamma\neq0;\\
	\overline{\eta}(-\alpha)(p-1)^{2}G\overline{G},& \text{if } \alpha\neq0 \text{ and }\beta=\gamma=0;\\
	-\overline{\eta}(-\alpha)(p-1)G\overline{G}, &\text{if } \alpha\neq0,~ \beta=0 \text{ and } \gamma\neq0;\\
	\Big(\overline{\eta}(-\beta)(p-1)-\overline{\eta}(-\alpha)\Big)(p-1)G\overline{G},&\text{if } \alpha\neq0,~ \beta\neq0 \text{ and } \alpha\beta=\gamma^{2};\\
	-\Big(\overline{\eta}(-\beta)+\overline{\eta}(-\alpha)\Big)(p-1)G\overline{G},&\text{if } \alpha\neq0,~ \beta\neq0 \text{ and } \alpha\beta\neq\gamma^{2}.
	\end{cases}~~~~~~~~~~~~~ $$
\end{lemma}
\begin{proof} By Lemmas \ref{p4le2} and \ref{p4le4}, we have
	\begin{align*}
	\Omega_{4}&=\sum_{d\in\mathbb{F}_{q}}\sum_{x\in\mathbb{F}_{p}^{*}}\zeta_{p}^{x\mathrm{Tr}(d^{2})}\sum_{y\in\mathbb{F}_{p}^{*}}\zeta_{p}^{y\mathrm{Tr}(ad)}\sum_{z\in\mathbb{F}_{p}^{*}}\zeta_{p}^{z\mathrm{Tr}(bd)}\\
	&=\sum_{x,y,z\in\mathbb{F}_{p}^{*}}\sum_{d\in\mathbb{F}_{q}}\chi_{1}(xd^{2}+(ya+zb)d)\\
	&=\sum_{x,y,z\in\mathbb{F}_{p}^{*}}\chi_{1}\Big(-\frac{(ya+zb)^{2}}{4x}\Big)\eta(x)G\\
	&=\begin{cases}
	G\sum\limits_{x,y,z\in\mathbb{F}_{p}^{*}}\chi_{1}\Big(-\frac{(ya+zb)^{2}}{4x}\Big),&\text{if } 2\mid m;\\
	G\sum\limits_{x,y,z\in\mathbb{F}_{p}^{*}}\chi_{1}\Big(-\frac{(ya+zb)^{2}}{4x}\Big)\overline{\eta}(x),&\text{if } 2\nmid m.
	\end{cases}
	\end{align*}
	For the sake of simplicity, we divide the remaining proof into two cases.\\
	\textbf{Case 1:} Suppose that $2\mid m$. Then
	\begin{align*}
	\Omega_{4}&=\sum\limits_{x,y,z\in\mathbb{F}_{p}^{*}}  \chi_{1}\Big(-\frac{(ya+zb)^{2}}{4x}\Big)G\\
	&=G\sum_{x,y,z\in\mathbb{F}_{p}^{*}}\chi_{1}\big(\frac{-y^{2}a^{2}-z^{2}b^{2}-2yzab}{4x}\big)\\
	&=G\sum_{x,y\in\mathbb{F}_{p}^{*}}\overline{\chi}_{1}\big(-\frac{y^{2}}{4x}\alpha\big)\sum_{z\in\mathbb{F}_{p}^{*}}\overline{\chi}_{1}\big(-\frac{z^{2}}{4x}\beta-\frac{yz}{2x}\gamma\big)\\
	&=\begin{cases}
	G(p-1)^{3},& \text{if } \alpha=0,~ \beta=0 \text{ and } \gamma=0;\\
	-G(p-1)^{2},& \text{if } \alpha=0,~ \beta=0 \text{ and } \gamma\neq0;\\
	-G(p-1)^{2},& \text{if } \alpha=0,~ \beta\neq0 \text{ and } \gamma=0;\\
	\overline{\eta}(-1)G\overline{G}^{2 }(p-1)-G(p-1)^{2},& \text{if } \alpha=0,~ \beta\neq0 \text{ and } \gamma\neq0;\\
	-G(p-1)^{2},& \text{if } \alpha\neq0,~ \beta=0 \text{ and } \gamma=0;\\
	G(p-1), &\text{if } \alpha\neq0,~ \beta=0 \text{ and } \gamma\neq0;\\
	G(p-1),&\text{if } \alpha\neq0,~ \beta\neq0 \text{ and } \alpha\beta=\gamma^{2}; \\
	\Big(\overline{\eta}\big(\alpha\beta-\gamma^{2}\big)\overline{G}^{2}+1\Big)G(p-1),&\text{if } \alpha\neq0,~ \beta\neq0 \text{ and } \alpha\beta\neq\gamma^{2}.
	\end{cases}
	\end{align*}
	\textbf{Case 2:} Assume that $2\nmid m$. Then
	\begin{align*}
	\Omega_{4}&=G\sum\limits_{x,y,z\in\mathbb{F}_{p}^{*}}\chi_{1}\Big(-\frac{(ya+zb)^{2}}{4x}\Big)\overline{\eta}(x)\\
	&=G\sum_{x,y,z\in\mathbb{F}_{p}^{*}}\chi_{1}\big(\frac{-y^{2}a^{2}-z^{2}b^{2}-2yzab}{4x}\big)\overline{\eta}(x)\\
	&=G\sum_{x,y\in\mathbb{F}_{p}^{*}}\overline{\chi}_{1}\big(-\frac{y^{2}}{4x}\alpha\big)\overline{\eta}(x)\sum_{z\in\mathbb{F}_{p}^{*}}\overline{\chi}_{1}\big(-\frac{z^{2}}{4x}\beta-\frac{yz}{2x}\gamma\big)\\
	&=\begin{cases}
	0,& \text{if } \alpha=0\text{ and } \beta=0;\\
	G\overline{G}\overline{\eta}(-\beta)(p-1)^{2},& \text{if } \alpha=0,~ \beta\neq0 \text{ and } \gamma=0;\\
	-G\overline{G}\overline{\eta}(-\beta)(p-1),& \text{if } \alpha=0,~ \beta\neq0 \text{ and } \gamma\neq0;\\
	\overline{\eta}(-\alpha)G\overline{G}(p-1)^{2},& \text{if } \alpha\neq0,~ \beta=0 \text{ and } \gamma=0;\\
	-\overline{\eta}(-\alpha)G\overline{G}(p-1), &\text{if } \alpha\neq0,~ \beta=0 \text{ and } \gamma\neq0;\\
	\Big(\overline{\eta}(-\beta)(p-1)-\overline{\eta}(-\alpha)\Big)G\overline{G}(p-1),&\text{if } \alpha\neq0,~ \beta\neq0 \text{ and } \alpha\beta=\gamma^{2} ;\\
	-\Big(\overline{\eta}(-\beta)+\overline{\eta}(-\alpha)\Big)G\overline{G}(p-1),&\text{if } \alpha\neq0,~ \beta\neq0 \text{ and } \alpha\beta\neq\gamma^{2}.
	\end{cases}
	\end{align*}
	Thus the result is established.
\end{proof}
\begin{lemma}\label{p4le17}
	 For $a,b\in\mathbb{F}_{q}^{*}$, where $a\neq wb$ for all $w\in\mathbb{F}_{p}^{*}$, let
	$$\Omega=\#\{d\in\mathbb{F}_{q}:\mathrm{Tr}(d^{2})=0,\mathrm{Tr}(ad)=0 \text{ and }\mathrm{Tr}(bd)=0\}.$$
	Then\\
	$1.$ if  $m$ is even, we have
	$$\Omega=\begin{cases}
	p^{m-3}+p^{-1}(p-1)G,& \text{if } \alpha=\beta=\gamma=0;\\
		p^{m-3},&\text{if } \alpha\neq0,~ \beta\neq0 \text{ and } \alpha\beta=\gamma^{2} ;\\
	p^{m-3}+\overline{\eta}\big(\gamma^{2}-\alpha\beta\big)p^{-2}(p-1)G,&\text{if } \alpha\neq0,~ \beta\neq0 \text{ and } \alpha\beta\neq\gamma^{2};\\
	p^{m-3}+p^{-2}(p-1)G,& \text{if } \gamma\neq0 \text{ and atmost one of } \alpha \text{ and } \beta \text{ is nonzero};\\
	p^{m-3},& \text{if } \gamma=0 \text{ and exactly one of } \alpha\text{ and } \beta \text{ is nonzero}.
	\end{cases}
	$$
	$2.$ if $m$ is odd, we have
	$$\Omega=\begin{cases}
	p^{m-3}+\overline{\eta}(-\alpha)p^{-2}(p-1)G\overline{G},& \text{if } \alpha\neq0 \text{ and }\beta=\gamma=0;\\
	p^{m-3},&\text{if } \alpha\neq0,~ \beta\neq0 \text{ and } \alpha\beta\neq\gamma^{2}\text {or }\alpha=\beta=0;\\
	p^{m-3},& \text{if } \gamma\neq0 \text{ and exactly one of }\alpha\text{ and }\beta \text{ is nonzero};\\
	p^{m-3}+\overline{\eta}(-\beta)p^{-2}(p-1)G\overline{G},&\text{if } \alpha\neq0,~ \beta\neq0 \text{ and } \alpha\beta=\gamma^{2}\text{ or }\beta\neq0\text{ and }\alpha=\gamma=0.
	\end{cases}$$
\end{lemma}
\begin{proof} By definition of additive character, we have
	\begin{align*}
	\Omega&=\frac{1}{p^{3}}\sum_{d\in\mathbb{F}_{q}}\Big(\sum_{x\in\mathbb{F}_{p}}\zeta_{p}^{x\mathrm{Tr}(d^{2})}\Big)\Big(\sum_{y\in\mathbb{F}_{p}}\zeta_{p}^{y\mathrm{Tr}(ad)}\Big)\Big(\sum_{z\in\mathbb{F}_{p}}\zeta_{p}^{z\mathrm{Tr}(bd)}\Big)\\
	&=\frac{1}{p^{3}}\sum_{d\in\mathbb{F}_{q}}\Big(\sum_{x\in\mathbb{F}_{p}}\zeta_{p}^{x\mathrm{Tr}(d^{2})}\Big)\Big(\sum_{y\in\mathbb{F}_{p}}\zeta_{p}^{y\mathrm{Tr}(ad)}\Big)\Big(1+\sum_{z\in\mathbb{F}_{p}^{*}}\zeta_{p}^{z\mathrm{Tr}(bd)}\Big)\\
	&=\frac{N_{a}}{p}+\frac{1}{p^{3}}(\Omega_{1}+\Omega_{2}+\Omega_{3}+\Omega_{4}),
	\end{align*}
	where
	\begin{align*}
	\Omega_{1}&=\sum_{d\in\mathbb{F}_{q}}\sum_{z\in\mathbb{F}_{p}^{*}}\zeta_{p}^{z\mathrm{Tr}(bd)}=\sum_{z\in\mathbb{F}_{p}^{*}}\sum_{d\in\mathbb{F}_{q}}\chi_{1}(zbd)=0,\\
	\Omega_{2}&=\sum_{d\in\mathbb{F}_{q}}\sum_{x\in\mathbb{F}_{p}^{*}}\zeta_{p}^{x\mathrm{Tr}(d^{2})}\sum_{z\in\mathbb{F}_{p}^{*}}\zeta_{p}^{z\mathrm{Tr}(bd)},\\
	\Omega_{3}&=\sum_{d\in\mathbb{F}_{q}}\sum_{y\in\mathbb{F}_{p}^{*}}\zeta_{p}^{y\mathrm{Tr}(ad)}\sum_{z\in\mathbb{F}_{p}^{*}}\zeta_{p}^{z\mathrm{Tr}(bd)}=\sum_{y\in\mathbb{F}_{p}^{*}}\sum_{z\in\mathbb{F}_{p}^{*}}\sum_{d\in\mathbb{F}_{q}}\chi_{1}((ya+zb)d)=0,\\
	\Omega_{4}&=\sum_{d\in\mathbb{F}_{q}}\sum_{x\in\mathbb{F}_{p}^{*}}\zeta_{p}^{x\mathrm{Tr}(d^{2})}\sum_{y\in\mathbb{F}_{p}^{*}}\zeta_{p}^{y\mathrm{Tr}(ad)}\sum_{z\in\mathbb{F}_{p}^{*}}\zeta_{p}^{z\mathrm{Tr}(bd)}.
	\end{align*}
	The result now follows from the Lemmas \ref{p4le6}, \ref{p4le7} and \ref{p4le16}.
\end{proof}

\section{Main results}\label{sec4}
Now, we are ready to prove our main results.
For the defining set $D^{*}=\{x\in\mathcal{R}_{m}^{*}: \mathrm{tr}(x^{2})=0\}$, it is complicated  to find weight distribution of the corresponding code $\mathcal{C}_{D^{*}}$ based on construction in \eqref{eq1}. So, we omit this case and  consider the defining set 
$$D=\{d\in\mathbb{F}_{q}^{*}: \mathrm{Tr}(d^{2})=0\}.$$
Let $D=\{d\in\mathbb{F}_{q}^{*}: \mathrm{Tr}(d^{2})=0\}$=$\{d_{1},d_{2},\ldots,d_{n}\}\subset\mathbb{F}_{q}$. Then, for $x=a+ub\in\mathcal{R}_{m}$, define 
\begin{align}
\mathcal{C}_{D}&=\{(\mathrm{tr}(xd_{1}),\mathrm{tr}(xd_{2})\ldots,\mathrm{tr}(xd_{n})):x\in\mathcal{R}_{m}\}\label{eq4}\\
&=\{(\mathrm{Tr}(ad_{1})+u\mathrm{Tr}(bd_{1}),\mathrm{Tr}(ad_{1})+u\mathrm{Tr}(bd_{2})\ldots,\mathrm{Tr}(ad_{n})+u\mathrm{Tr}(bd_{n})):a,b\in\mathbb{F}_{q}\}\nonumber,
\end{align}
and
$$\phi(\mathcal{C}_{D})=\{(\mathrm{Tr}(ad_{1}),\mathrm{Tr}(ad_{2}),\ldots,\mathrm{Tr}(ad_{n}),\mathrm{Tr}(bd_{1}),\mathrm{Tr}(bd_{2}),\ldots,\mathrm{Tr}(bd_{n})):a, b\in\mathbb{F}_{q}\}.$$

Depending on whether $2$ divides $m$ or not, we present the Hamming-weight distribution of  $\mathcal{C}_{D}$ employing symplectic weight distribution of $\phi(\mathcal{C}_{D})$ in the following Theorems.
\begin{theorem}\label{p4theorem1}
	Let $m>2$ be even. Then the introduced linear code  $\mathcal{C}_{D}$ has parameters $[p^{m-1}+p^{-1}(p-1)G-1,m]$  with the weight distribution of Table 3.
\end{theorem}
\begin{center}
	\textbf{Table 3:}
	The weight distribution of the code $\mathcal{C}_{D}$ in Theorem \ref{p4theorem1}\\
	\begin{tabular}{|p{4.5cm}|p{9.5cm}|}
		\hline
		Weight w & Multiplicity $A_{w}$\\
		\hline
		$0$ & $1$\\
		\hline
		$p^{m-3}(p^{2}-1)$ & $p^{2m-3}-p^{m-3}(p^{3}-p^{2}+3p-1)+p^{-1}(p^{2}-1)(p^{m-2}-1)G+p^{-3}(p-1)^{3}G^{2}+p$\\
		\hline
		$p^{m-3}(p^{2}-1)+p^{-2}(p-1)^{2}G$ & $\frac{1}{2}(p-1)\big(p^{2m-2}(p+1)+p^{m-2}(p^{2}-1)G-2p^{m-2}(2p-1)-p^{-2}(p-1)(p-2)G^{2}\big)$\\
		\hline
		$p^{m-3}(p^{2}-1)+p^{-1}(p-1)G$ & $(p-1)\big(p^{2m-3}(p+1)-p^{m-3}(p^{3}-p^{2}+3p-1)-p^{-3}(2p^{2}-3p+1)G^{2}-p^{-1}(p+1)(p^{m-2}-1)G\big)$\\
		\hline
		$(p^{2}-1)(p^{m-3}+p^{-2}G)$ & $\frac{1}{2}(p-1)^{2}(p^{2m-2}-p^{m-2}(p+1)G+p^{-1}G^{2})$\\
		\hline
		$p^{m-2}(p-1)$ & $(p+1)(p^{m-1}+p^{-1}(p-1)G-1)$\\
		\hline
		$(p-1)(p^{m-2}+p^{-1}G)$ & $(p^{2}-1)(p^{m-1}-p^{-1}G)$\\
		\hline
	\end{tabular}
\end{center}
\begin{proof} Recall that $\alpha=\mathrm{Tr}(a^{2})$, $\beta=\mathrm{Tr}(b^{2})$ and $\gamma=\mathrm{Tr}(ab)$, where $a,b\in\mathbb{F}_{q}^{*}.$ For any $a,b\in\mathbb{F}_{q}$, we define
	$\aleph(a,b)=\#\{d\in\mathbb{F}_{q}:\mathrm{Tr}(d^{2})=0,\mathrm{Tr}(ad)=0 \text{ and }\mathrm{Tr}(bd)=0\}.$
	Then, from Lemmas \ref{p4le5}, \ref{p4le7} and \ref{p4le17}, one can easily obtain that 
	$$\aleph(a,b)=\begin{cases}
	\Omega,&\text{if }a,b\in\mathbb{F}_{q}^{*} \text{ and } a\neq wb \text{ for all }w\in\mathbb{F}_{p}^{*};\\
	N_{a},&\text{if } a\neq0 \text{ and }b=0;\\
	N_{b},&\text{if } b\neq0 \text{ and }a=0 \text{ or }  \text{ for }a,b\in\mathbb{F}_{q}^{*}, a=wb\text{ for some }w\in\mathbb{F}_{p}^{*};\\
	n_{0},&\text{if }a=b=0.
	\end{cases}$$ 
	The distinct values of $\aleph(a,b)$ under different conditions are described in the following table.
\begin{center}
	\begin{tabular}{|p{4.5cm}|p{9.5cm}|}
		\hline
		\text{The value of }$\aleph(a,b)$& \text{Conditions}\\
		\hline
		 $p^{m-3}+p^{-1}(p-1)G$ & $\alpha=\beta=\gamma=0$ \\
		 \hline
		$p^{m-3}+p^{-2}(p-1)G$ & $\alpha\neq0,~\beta\neq0 \text{ and }\overline{\eta}(\gamma^{2}-\alpha\beta)=1$ or $\gamma\neq0$ and atmost one of $\alpha$ and $\beta$ is nonzero \\
		\hline
		$p^{m-3}$ & $\alpha\neq0,~\beta\neq0 \text{ and } \alpha\beta=\gamma^{2}$ or $\gamma=0$ and exactly one of $\alpha$ and $\beta$ is nonzero \\
		\hline
		$p^{m-3}-p^{-2}(p-1)G$& $\alpha\neq0,~\beta\neq0 \text{ and }\overline{\eta}(\gamma^{2}-\alpha\beta)=-1$  \\
		\hline
		$p^{m-2}+p^{-1}(p-1)G$& $\alpha=0$ or $\beta=0$ \\
		\hline
		$p^{m-2}$ &$\alpha\neq0$ or $\beta\neq0$ \\
		\hline
		$p^{m-1}+p^{-1}(p-1)G$ & $a=b=0$\\
		\hline
		\end{tabular}
\end{center}
Let $\textbf{c}\neq0\in\mathcal{C}_{D}$. Then the symplectic-weight $\text{swt}(\phi(\textbf{c}))$ of $\phi(\textbf{c})$ is given by
$$\text{swt}(\phi(\textbf{c}))=n_{0}-\aleph(a,b),$$
and, by Lemma \ref{p4le4} and previous table, has the following six possible values:
\begin{enumerate}[label=(\roman*),leftmargin=*, widest=iii]
	\item $p^{m-3}(p^{2}-1)+p^{-2}(p-1)^{2}G$;
	\item $p^{m-3}(p^{2}-1)$;
	\item $p^{m-3}(p^{2}-1)+p^{-1}(p-1)G$;
	\item $(p^{2}-1)(p^{m-3}+p^{-2}G)$;
	\item $p^{m-2}(p-1)$;
	\item $(p-1) (p^{m-2}+p^{-1}G)$.
\end{enumerate}
 For the sake of simplicity, the number $\mathcal{N}_{i}~(i=1,2,3,4,5,6,7)$ satisfying the related condition is denoted by the following table.
\begin{center}
	\begin{tabular}{|p{10cm}|p{4.5cm}|}
		\hline
		\text{Conditions}&\text{Number of pairs $(a,b)$}\\ &\text{satisfying conditions }\\
		\hline
		 $\alpha=\beta=\gamma=0$ & $\mathcal{N}_{1}$\\ 
		\hline
		 $\alpha\neq0,~\beta\neq0 \text{ and }\overline{\eta}(\gamma^{2}-\alpha\beta)=1$ or $\gamma\neq0$ and atmost one of $\alpha$ and $\beta$ is nonzero  & $\mathcal{N}_{2}$\\
		\hline
		 $\alpha\neq0,~\beta\neq0 \text{ and } \alpha\beta=\gamma^{2}$ or $\gamma=0$ and exactly one of $\alpha$ and $\beta$ is nonzero &$\mathcal{N}_{3}$\\
		\hline
		$\alpha\neq0,~\beta\neq0 \text{ and }\overline{\eta}(\gamma^{2}-\alpha\beta)=-1$  &$\mathcal{N}_{4}$\\
		\hline
		 $\alpha=0$ or $\beta=0$ &$\mathcal{N}_{5}$\\
		\hline
		$\alpha\neq0$ or $\beta\neq0 $ & $\mathcal{N}_{6}$\\
		\hline
	 $a=b=0$&$\mathcal{N}_{7}$\\
		\hline
	\end{tabular}
\end{center}
Let $c,c_{i}\in\mathbb{F}_{p}^{*}$ ($1\leq i\leq3$). Then, from Lemmas \ref{p4le5} and \ref{p4le15}, one can easily acquire that
{\small\begin{align*}
\mathcal{N}_{1}&=N(0,0,0)-\Big(2(n_{0}-1)+(p-1)(n_{0}-1)+1\Big)=N(0,0,0)-\big((p+1)n_{0}-p\big),\\
\mathcal{N}_{2}&=\sum_{\overline{\eta}(c_{3}^{2}-c_{1}c_{2})=1}N(c_{1},c_{2},c_{3})+\sum_{\overline{\eta}(-c_{1}c_{2})=1}N(c_{1},c_{2},0)+(p-1)N(0,0,c_{3})+2(p-1)^{2}N(c_{1},0,c_{3}),\\
\mathcal{N}_{3}&=\sum_{c_{3}^{2}=c_{1}c_{2}}N(c_{1},c_{2},c_{3})+2(p-1)N(c_{1},0,0)-(p^{2}-1)n_{c},\\
\mathcal{N}_{4}&=\sum_{\overline{\eta}(c_{3}^{2}-c_{1}c_{2})=-1}N(c_{1},c_{2},c_{3})+\sum_{\overline{\eta}(-c_{1}c_{2})=-1}N(c_{1},c_{2},0),~
\mathcal{N}_{5}=(p+1)(n_{0}-1),\\
\mathcal{N}_{6}&=(p^{2}-1)n_{c},~
\mathcal{N}_{7}=1.
\end{align*}}\\
Applications of Lemmas \ref{p4le5}, \ref{p4le9} and \ref{p4le15} lead us to the complete proof of the theorem.
\end{proof}

\begin{example}
	Let $p=3~and~m=4$. Then the linear code $\mathcal{C}_{D}$ has the parameters  $[20, 4, 12]$ with Hamming-weight enumerator $1+240z^{12}+2160z^{16}+2000z^{18}+2160z^{20}$. The  code $\mathcal{C}_{D}$ is four weight linear code if $p=3$ and $m=4$.   
\end{example}
\begin{example}
	Let $p=3~and~m=6$. Then the linear code $\mathcal{C}_{D}$ has the parameters  $[260, 6, 162]$ with Hamming-weight enumerator $1+1040z^{162}+1872z^{180}+24960z^{216}+252720z^{228}+149760z^{234}+101088z^{240}$.  
\end{example}
\begin{theorem}\label{p4theorem2}
	Let $m\geq3$ be odd. Then the introduced linear code $\mathcal{C}_{D}$ has parameters $[p^{m-1}-1,m]$  with the weight distribution of Table 4.
\end{theorem}
\begin{center}
	\textbf{Table 4:}
	The weight distribution of the code $\mathcal{C}_{D}$ in Theorem \ref{p4theorem2}\\
	\begin{tabular}{|p{4.5cm}|p{9.5cm}|}
		\hline
		Weight w & Multiplicity $A_{w}$\\		
		\hline
		$0$ & $1$\\
		\hline
		$p^{m-3}(p^{2}-1)-p^{-2}(p-1)G\overline{G}$ & $\frac{1}{2}(p-1)(p^{2m-3}(p+1)+p^{m-3}(2p-1)(p-1)+p^{m-2}(p+1)G\overline{G}+p^{-4}(p-1)(p^{2}-p+1)G^{2}\overline{G}^{2})-\frac{1}{2}(p^{2}-1)(p^{m-1}+p^{-1}G\overline{G})$\\
		\hline
		$p^{m-3}(p^{2}-1)+p^{-2}(p-1)G\overline{G}$ & $\frac{1}{2}(p-1)(p^{2m-3}(p+1)+p^{m-3}(2p-1)(p-1)-p^{m-2}(p+1)G\overline{G}+p^{-4}(p-1)(p^{2}-p+1)G^{2}\overline{G}^{2})-\frac{1}{2}(p^{2}-1)(p^{m-1}-p^{-1}G\overline{G})$\\
		\hline
		$p^{m-3}(p^{2}-1)$ & $(p-1)^{2}(p^{2m-3}(p+1)-p^{m-3}(2p-1)-p^{-4}(p^{2}-p+1)G^{2}\overline{G}^{2})+p^{2m-2}-(p+1)p^{m-1}+p$\\
		\hline
		$p^{m-2}(p-1)-p^{-2}(p-1)G\overline{G}$ & $\frac{1}{2}(p^{2}-1)(p^{m-1}+p^{-1}G\overline{G})$\\
		\hline
		$p^{m-2}(p-1)+p^{-2}(p-1)G\overline{G}$ & $\frac{1}{2}(p^{2}-1)(p^{m-1}-p^{-1}G\overline{G})$\\
		\hline
		$p^{m-2}(p-1)$ & $(p+1)(p^{m-1}-1)$\\
		\hline		
	\end{tabular}
\end{center}
\begin{proof}  Let $\alpha,\beta$ and $\gamma$ have the same meanings as before. For each pair $(a,b)\in\mathbb{F}_{q}\times\mathbb{F}_{q}$, we define
	$\aleph(a,b)=\#\{d\in\mathbb{F}_{q}:\mathrm{Tr}(d^{2})=0,\mathrm{Tr}(ad)=0 \text{ and }\mathrm{Tr}(bd)=0\}.$
	Then, from Lemmas \ref{p4le5}, \ref{p4le7} and \ref{p4le17}, one can easily obtain that \\
	$$\aleph(a,b)=\begin{cases}
\Omega,&\text{if }a,b\in\mathbb{F}_{q}^{*} \text{ and } a\neq wb \text{ for all }w\in\mathbb{F}_{p}^{*};\\
N_{a},&\text{if } a\neq0 \text{ and }b=0;\\
N_{b},&\text{if } b\neq0 \text{ and }a=0 \text{ or }  \text{ for }a,b\in\mathbb{F}_{q}^{*}, a=wb\text{ for some }w\in\mathbb{F}_{p}^{*};\\
n_{0},&\text{if }a=b=0.
\end{cases}$$ 
	The distinct values of $\aleph(a,b)$ under different conditions are described in the following table.
	\begin{center}
		\begin{tabular}{|p{4.5cm}|p{9.5cm}|}
			\hline
			\text{The value of }$\aleph(a,b)$& \text{Conditions}\\
			\hline
			$p^{m-3}+p^{-2}(p-1)G\overline{G}$&$\overline{\eta}(-\alpha)=1$ and $\beta=\gamma=0$ or $\alpha\neq0,\overline{\eta}(-\beta)=1$ and $\alpha\beta=\gamma^{2}$ or $\overline{\eta}(-\beta)=1$ and $\alpha=\gamma=0$\\
			\hline
			$p^{m-3}-p^{-2}(p-1)G\overline{G}$&$\overline{\eta}(-\alpha)=-1$ and $\beta=\gamma=0$ or $\alpha\neq0,\overline{\eta}(-\beta)=-1$ and $\alpha\beta=\gamma^{2}$ or $\overline{\eta}(-\beta)=-1$ and $\alpha=\gamma=0$\\
			\hline
			$p^{m-3}$ & $\alpha\neq0,~\beta\neq0 \text{ and } \alpha\beta\neq\gamma^{2}$ or $\alpha=\beta=\gamma=0$ or $\alpha=\beta=0 \text{ and }\gamma\neq0$ or $\gamma\neq0$ and exactly one of $\alpha$ and $\beta$ is nonzero\\
			\hline		
				$p^{m-2}+p^{-2}(p-1)G\overline{G}$& $\overline{\eta}(-\alpha)=1$ or $\overline{\eta}(-\beta)=1$\\
				\hline
					$p^{m-2}-p^{-2}(p-1)G\overline{G}$& $\overline{\eta}(-\alpha)=-1$ or $\overline{\eta}(-\beta)=-1$\\
				\hline	
			$p^{m-2}$& $ \alpha=0$ or $\beta=0$ \\
		\hline
			$p^{m-1}$&$a=b=0$\\
			\hline
		\end{tabular}
	\end{center}
Let $\textbf{c}\neq0\in\mathcal{C}_{D}$. Then the symplectic-weight $\text{swt}(\phi(\textbf{c}))$ of $\phi(\textbf{c})$ is given by
$$\text{swt}(\phi(\textbf{c}))=n_{0}-\aleph(a,b),$$
and, by Lemma \ref{p4le4} and previous table, has the following six possible values:
\begin{enumerate}[label=(\roman*),leftmargin=*, widest=iii]
	\item $p^{m-3}(p^{2}-1)-p^{-2}(p-1)G$;
	\item $p^{m-3}(p^{2}-1)+p^{-2}(p-1)G$;
	\item $p^{m-3}(p^{2}-1)$;
	\item $p^{m-2}(p^{2}-1)-p^{-2}(p-1)G$;
	\item $p^{m-2}(p^{2}-1)+p^{-2}(p-1)G$;
	\item $p^{m-2}(p-1)$.
\end{enumerate}
 The number $\mathcal{N}_{i}~(i=1,2,3,4,5)$ satisfying the related condition is denoted by the following table.
\begin{center}
	\begin{tabular}{|p{10cm}|p{4cm}|}
		\hline
		\text{Conditions}&\text{Number of pairs $(a,b)$}\\ &\text{satisfying conditions }\\
		\hline
	$\overline{\eta}(-\alpha)=1$ and $\beta=\gamma=0$ or $\alpha\neq0,\overline{\eta}(-\beta)=1$ and $\alpha\beta=\gamma^{2}$ or $\overline{\eta}(-\beta)=1$ and $\alpha=\gamma=0$&$\mathcal{N}_{1}$\\
		\hline
		$\overline{\eta}(-\alpha)=-1$ and $\beta=\gamma=0$ or $\alpha\neq0,\overline{\eta}(-\beta)=-1$ and $\alpha\beta=\gamma^{2}$ or $\overline{\eta}(-\beta)=-1$ and $\alpha=\gamma=0$ &  $\mathcal{N}_{2}$\\
		\hline
		$\alpha\neq0,~\beta\neq0 \text{ and } \alpha\beta\neq\gamma^{2}$ or $\alpha=\beta=\gamma=0$ or $\alpha=\beta=0 \text{ and }\gamma\neq0$ or $\gamma\neq0$ and exactly one of $\alpha$ and $\beta$ is nonzero  &$\mathcal{N}_{3}$\\
		\hline
		$\overline{\eta}(-\alpha)=1$ \text{ or } $\overline{\eta}(-\beta)=1$ & $\mathcal{N}_{4}$\\
		\hline
		$\overline{\eta}(-\alpha)=-1$ \text{ or } $\overline{\eta}(-\beta)=-1$ & $\mathcal{N}_{5}$\\
		\hline
		$\alpha=0$ or $\beta=0$ & $\mathcal{N}_{6}$\\
				\hline
		$a=b=0$& $\mathcal{N}_{7}$\\
		\hline
	\end{tabular}
\end{center}
\begin{small}
\begin{align*}
\mathcal{N}_{1}&=\sum_{\overline{\eta}(-c_{1})=1}N(c_{1},0,0)+\sum_{\overline{\eta}(-c_{2})=1,c_{1}c_{2}=c_{3}^{2}}N(c_{1},c_{2},c_{3})+\sum_{\overline{\eta}(-c_{2})=1}N(0,c_{2},0)-(p+1)\sum_{\overline{\eta}(-c)=1}n_{c},\\
\mathcal{N}_{2}&=\sum_{\overline{\eta}(-c_{1})=-1}N(c_{1},0,0)+\sum_{\overline{\eta}(-c_{2})=-1,c_{1}c_{2}=c_{3}^{2}}N(c_{1},c_{2},c_{3})+\sum_{\overline{\eta}(-c_{2})=-1}N(0,c_{2},0)-(p+1)\sum_{\overline{\eta}(-c)=-1}n_{c},\\
\mathcal{N}_{3}&=\sum_{c_{1}c_{2}\neq c_{3}^{2}}N(c_{1},c_{2},c_{3})+3(p-1)^{2}N(c_{1},c_{2},0)+(p-1)N(0,0,c_{3})+N(0,0,0)-\big((p+1)n_{0}-p\big),\\
\mathcal{N}_{4}&=(p+1)\sum_{\overline{\eta}(-c)=1}n_{c},~
\mathcal{N}_{5}=(p+1)\sum_{\overline{\eta}(-c)=-1}n_{c},~
\mathcal{N}_{6}=(p+1)(n_{0}-1),~
\mathcal{N}_{7}=1.
\end{align*}
\end{small}	
The proof of the theorem follows by Lemmas \ref{p4le5}, \ref{p4le9} and \ref{p4le15}.
\end{proof}
\begin{example}
	Let $p=m=3$. Then the linear code $\mathcal{C}_{D}$ has the parameters  $[8, 3, 4]$ with the Hamming-weight enumerator $1+48 z^{4}+224z^{6}+456z^{8}$. The linear code $\mathcal{C}_{D}$ is three-weight linear code  if  p=m=3.\end{example}

\begin{example}
	Let $p=3$  and $m=5$. Then the linear code $\mathcal{C}_{D}$ has the parameters  $[80, 5, 54]$ with the Hamming-weight enumerator $1+360z^{48}+320z^{54}+288z^{60}+11520z^{66}+40800z^{72}+5760z^{78}$.\end{example}
\section{Conclusions}\label{sec5}
In the present paper, we have constructed a family of a few weight linear codes based on the defining set. Moreover, the  weight distributions of the linear codes have been determined by employing symplectic-weight distributions of their Gray images. We have used a relatively new approach to determine Hamming-weight distributions of linear codes over $\mathbb{F}_{p}+u\mathbb{F}_{p}$. In the same manner, more codes may be constructed over such rings based on defining sets, and we leave this for future research work.

\textbf{Acknowledgements.} This research work was done as a part of Pavan Kumar's Ph.D. thesis with the affiliation Aligarh Muslim University, and was supported by the University Grants Commission, New Delhi, India, under  JRF in Science, Humanities $\&$ Social Sciences scheme with Grant number  11-04-2016-413564.

\end{document}